\documentclass{ieee-ojvt}
\usepackage{amsmath,amsfonts,amsthm,mathrsfs}
\usepackage{algorithmic}
\usepackage{algorithm}
\usepackage{array}
\usepackage[caption=false,font=normalsize,labelfont=sf,textfont=sf]{subfig}
\usepackage{textcomp}
\usepackage{stfloats}
\usepackage{url}
\usepackage{verbatim}
\usepackage{graphicx}
\usepackage{cite}
\usepackage{makecell}
\usepackage{graphicx}
\usepackage{booktabs}
\usepackage{bm}
\usepackage{comment}

\newtheorem{theorem}{Theorem}

\newtheorem{definition}{Definition}
\newtheorem{assumption}{Assumption}
\def\BibTeX{{\rm B\kern-.05em{\sc i\kern-.025em b}\kern-.08em
    T\kern-.1667em\lower.7ex\hbox{E}\kern-.125emX}}
\AtBeginDocument{\definecolor{ojcolor}{cmyk}{0.93,0.59,0.15,0.02}}
\AtBeginDocument{\definecolor{ojcolor}{cmyk}{0.93,0.59,0.15,0.02}}
\def\OJlogo{\vspace{-12pt}\includegraphics[height=24pt]{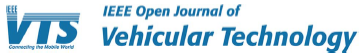}}
\DeclareRobustCommand{\IEEEauthorrefmark}[1]{\smash{\textsuperscript{\footnotesize #1}}}

\begin{document}

\receiveddate{2 July, 2025}
\reviseddate{XX July, 2025} 

\title{RALLY: Role-Adaptive LLM-Driven Yoked Navigation for Agentic UAV Swarms}

\author{Ziyao Wang\authorrefmark{1,3}, 
Rongpeng Li\authorrefmark{2} (Senior Member, IEEE),
Sizhao Li\authorrefmark{2},
Yuming Xiang\authorrefmark{2},
Haiping Wang\authorrefmark{3},
Zhifeng Zhao\authorrefmark{3,2} (Senior Member, IEEE),
and Honggang Zhang\authorrefmark{4} (Fellow, IEEE)}
\IEEEspecialpapernotice{(Invited Paper)}

\affil{Hangzhou Institute for Advanced Study, UCAS, Hangzhou 311500, China}
\affil{Zhejiang University, Hangzhou 310027, China}
\affil{Zhejiang Lab, Hangzhou 310012, China}
\affil{Macau University of Science and Technology, Macau, China}
\corresp{Emails: wangziyao23@mails.ucas.ac.cn, \{lirongpeng; liszh5; xiangym1999\}@zju.edu.cn, \{wanghp, zhaozf\}@zhejianglab.org, honggang.zhang@ieee.org. Corresponding Author: Rongpeng Li.}
\markboth{RALLY: Role-Adaptive LLM-Driven Yoked Navigation for Agentic UAV Swarms}{Z. Wang \textit{et al.}}


\begin{abstract}
Intelligent control of Unmanned Aerial Vehicles (UAVs) swarms has emerged as a critical research focus, and it typically requires the swarm to navigate effectively while avoiding obstacles and achieving continuous coverage over multiple mission targets. 
Although traditional Multi-Agent Reinforcement Learning (MARL) approaches offer dynamic adaptability, they are hindered by the semantic gap in black-boxed communication and the rigidity of homogeneous role structures, resulting in poor generalization and limited task scalability.
Recent advances in Large Language Model (LLM)-based control frameworks demonstrate strong semantic reasoning capabilities by leveraging extensive prior knowledge. 
However, due to the lack of online learning and over-reliance on static priors, these works often struggle with effective exploration, leading to reduced individual potential and overall system performance. 
To address these limitations, we propose a Role-Adaptive LLM-Driven Yoked navigation algorithm RALLY. 
Specifically, we first develop an LLM-driven semantic decision framework that uses structured natural language for efficient semantic communication and collaborative reasoning. 
Afterward, we introduce a dynamic role-heterogeneity mechanism for adaptive role switching and personalized decision-making. 
Furthermore, we propose a Role-value Mixing Network (RMIX)-based assignment strategy that integrates LLM offline priors with MARL online policies to enable offline training of role selection strategies.
Experiments in the Multi-Agent Particle Environment (MPE) and a Software-In-The-Loop (SITL) platform demonstrate that RALLY outperforms conventional approaches in terms of task coverage, convergence speed, and generalization, highlighting its strong potential for collaborative navigation in agentic multi-UAV systems.
\end{abstract}
\maketitle

\section{Introduction}
\IEEEPARstart{N}{owadays}, Unmanned Aerial Vehicles (UAVs) have demonstrated significant potential in multi-agent pursuit-evasion game (e.g., disaster responses \cite{uzakov2020uav}). Typically,  the UAV swarm shall have the capability to avoid pursuing enemies and environmental obstacles while moving towards multiple target areas in certain formations, which is often known as a Dynamic Swarm coordination with Cooperative Evasion and Formation Coverage (DS-CEFC) task \cite{CIHRL2025}. 
By enabling UAVs to adjust their functions in response to real-time environmental conditions, dynamic role adaptation often leads to optimized coordination, improved task coverage, and enhanced robustness in complex and unpredictable scenarios \cite{Xia2023DynamicRD,Jin2024MultiAgentSL}.
Nevertheless, the underlying difficulty in coordinating roles and decision-making across multiple agents in swarms poses significant challenges. 
For example, traditional control-based algorithms \cite{vasarhelyi2014outdoor,lizzio2022review} are contingent on fully connected topologies and lack the required adaptability and scalability to large-scale dynamic environments, thus degrading the practicality in the real world \cite{yan2022relative}. 
Meanwhile, despite the incorporation of inter-agent communications \cite{lu_rethinking_2023,lu_selfcritical_2025} and cooperation \cite{sunehag2017value,das2020tarmac}, it still remains inevitable for decentralized Multi-Agent Reinforcement Learning (MARL) to yield conflicting roles and decisions  \cite{tan1993multi,oroojlooy2023review}. 
Given the inherent semantic reasoning capabilities and pretrained experience of Large Language Models (LLMs) \cite{xi2023rise,javaid2024large}, they offer a promising alternative to mitigate the critical issues of conflicting roles and inconsistent decision-making that plague purely MARL-based approaches.
\subsection{Related Works}
\subsubsection{MARL-Based UAV Swarm Control}
Deep Reinforcement Learning (DRL)\cite{mnih2015human} has significantly enhanced agent adaptability in complex tasks. 
However, in Multi-Agent Systems (MAS), critical challenges, such as environmental non-stationarity, rapidly expanding state spaces, and difficulty in credit assignment, hinder the ability of traditional methods to learn effective cooperative policies. 
To address this problem, the Centralized Training with Decentralized Execution (CTDE) paradigm is introduced with exemplary algorithms like MADDPG\cite{MADDPG} and  MAPPO\cite{MAPPO}. During training, it also involves techniques such as policy distillation\cite{green2019distillation} and imitation learning\cite{2017Imitation} to improve coordination in complex obstacle-laden environments. 
Nevertheless, these approaches struggle to assess an individual agent's contribution, as they typically optimize a global value function while neglecting the importance of localized utility. 
Value decomposition methods (e.g., VDN\cite{sunehag2017value}, QMIX\cite{QMIX}, QTRAN\cite{son2019qtran}) ameliorate this issue by decomposing the joint value function, thereby enabling analysis of individuals' contributions to cooperative decision-making. Furthermore, more advanced attention mechanisms or neural communication protocols (e.g., TarMAC\cite{das2020tarmac}, IMANet\cite{zou2022cooperative}, DAACMP\cite{mao2020learning}) can boost the effectiveness of filter messages. 
However, the direct communication of numerical vectors\cite{das2020tarmac,zou2022cooperative,mao2020learning}, which lacks interpretability and cannot convey task semantics, leads to information redundancy and bandwidth bottlenecks \cite{lu_rethinking_2023,lu_selfcritical_2025,lu_semanticsempowered_2024,ren_separate_2025}, and greatly limits algorithm generalization. 
Although the leader–follower architecture\cite{Alam2024LeaderFollowerFS} enables role differentiation, its static role assignments lack the flexibility to adapt to dynamic environmental conditions and varying formation sizes.
Therefore, researchers resort to integrating hierarchical control with consensus inference\cite{pateria2021hierarchical,cheng2022multi,CIHRL2025}, so as to simplify the inference interpretability and boost convergence speed. However, agents in these methods commonly remain homogeneous, and the policy networks learned via CTDE cannot inherently leverage the advantages of the natural heterogeneity in swarms. 
Therefore, building a scalable, efficient, dynamically adaptive, and interpretable heterogeneous multi-agent collaboration mechanism for UAV swarm control remains an open problem.
\subsubsection{LLM-Based Multi-Agent Systems}
LLMs have demonstrated near-human performance\cite{xi2023rise} in complex reasoning and planning tasks, providing new impetus for environment understanding and decision explainability in UAV swarms\cite{javaid2024large}. 
Leveraging vast prior knowledge, LLMs can not only be used for single-agent path planning (e.g., CoNavGPT\cite{yu2023co}, RoCo\cite{mandi2024roco}), but also facilitate high-level communication with low-level trajectory planning, significantly improving task efficiency, adaptability, and generalization\cite{du2023improving,liang2023encouraging,chan2023chateval}. Moreover, in MetaGPT\cite{hong2023metagpt}, CAMEL\cite{li2303camel} and ChatDev\cite{qian2023communicative}, LLMs can decompose complex missions into a number of subtasks handled collaboratively by different agents, thereby reducing ``hallucinations" \cite{zhang2023siren} and enhancing the ability to solve complex problems.
More importantly, these LLM-driven Multi-Agent (LLM-MA) systems\cite{zeng2024persllm,chuang2023simulating,chen2023multi} often customize LLMs into specialized or personalized roles. Therefore, through multi-agent collaboration, they replicate human-like collective intelligence and further enhance overall situational understanding and decision explainability, making inter-agent interactions more meaningful. However, LLM-based decision-making still heavily relies on prior knowledge and lacks exploration, often getting stuck in local optima\cite{Eigner2024DeterminantsOL}. Additionally, these methods leverage individual memory and self-evolution mechanisms to optimize agents independently, while neglecting potential synergistic effects arising from multi-agent interactions.
Furthermore, applying LLM-MA to build global consensus for DS-CEFC remains scarce, though some studies\cite{chen2023multi} start to focus on consensus negotiation among agents in a simplified scenario. 
On the other hand, some approaches\cite{zeng2024persllm} utilize diverse ``personalities", empowered by the LLM, for creative collaboration; nevertheless, the underlying fixed definitions of roles still struggle to accommodate dynamic task switching in DS-CEFC. 
In other words, enabling dynamic role adaptation with balanced exploration and exploitation remains a critical challenge for DS-CEFC.
\subsubsection{Integration of MARL and LLM}
There has been some light shed on deeply integrating MARL and LLMs\cite{Sun2024LLMbasedMR}. 
Prominently, semantic capabilities in LLMs can be leveraged as natural language interfaces to bridge human feedback and MAS \cite{Siedler2025LLMMediatedGO}, but the consensus inference therein often relies more on human supervision than on autonomous collaboration. 
Other studies distill LLM knowledge into smaller executable models or empower LLMs with human-in-the-loop feedback for policy generation, so as to accelerate MARL training and improve performance in complex environments \cite{Zhang2023BuildingCE,Sumers2023CognitiveAF,10577381}. Nevertheless, these methods often depend on offline human annotations and feedback, making them ill-suited for dynamic real-time changes; meanwhile, agent roles remain homogeneous, with only the number of agents being scaled up.
Some recent work treats LLMs as the core of heterogeneous agents in MARL\cite{xi2023rise}, assigning different ``personalities'' to agents. But these roles tend to be fixed and cannot adapt to environmental changes in DS-CEFC. 
Additionally, deploying LLMs in UAV swarms significantly strains computational resources and energy, limiting their practical use \cite{Yuan2024YOLOMARLYO}. 
Thus, it is meaningful to calibrate superior means to integrate MARL and LLMs, thus better harnessing heterogeneous swarm intelligence for improved generalization and adaptability in DS-CEFC.
\subsection{Contributions}
To accomplish consensus inference of multiple UAVs with heterogeneous roles for DS-CEFC, 
we propose a novel LLM-MARL-integrated framework called RALLY (Role-Adaptive LLM-Driven Yoked navigation). 
Built on our previous work\cite{CIHRL2025}, which unifies target selection, obstacle-avoidance navigation, and flight-control execution, RALLY serves as a significantly enhanced high-level consensus inference module towards establishing role-adaptive cooperative navigation. 
Specifically, RALLY comprises two core modules: an LLM-based two-stage semantic reasoning module and a Role-value Mixing Network (RMIX)-based credit-distribution mechanism.
The integration of semantic reasoning in LLM and policy learning in MARL makes it qualified for coordinating roles and decision-making across swarms for DS-CEFC. 
The main contributions of this work can be summarized as follows:
\begin{itemize}
\item We introduce RALLY, which consists of an LLM-driven semantic reasoning architecture for goal inference alongside a RMIX-based credit-distribution mechanism for role selection. RALLY accelerates consensus formation, improves convergence speed, and optimizes cooperative behaviors within the UAV swarm to fulfill DS-CEFC.
\item We implement a two-stage LLM-based semantic reasoning of intention and consensus inference. Replacing traditional numerical vector-based methods with more interpretable text contributes to the consensus inference efficiency.
\item Unlike static role assignment, RMIX dynamically assigns agents' roles during cooperative navigation.
By integrating offline LLM experiences with online MARL training, RMIX optimizes credit assignment and accelerates group coordination across diverse scenarios. 
\item To meet the distributed deployment and parallel inference demands on edge devices, we propose a capacity-migration algorithm that significantly reduces runtime memory footprint. 
We validate RALLY in the Multi-Particle Environment (MPE)\cite{lowe2017multi} and the Software-In-The-Loop (SITL) platform\cite{DangNguyen2019VisionBasedSF} based on Gazebo-ROS-PX4. 
Experimental results demonstrate that RALLY outperforms the latest MARL\cite{CIHRL2025} and pure LLM-driven approaches\cite{yu2023co,ditto2024large} in terms of task completion, convergence speed, generalization, and interpretability.
\end{itemize}
\subsection{Paper Structure}
The remainder of this paper is organized as follows. 
Section \ref{sec:model} introduces the system model and formulates the problem.  
Section \ref{sec:rally} outlines the algorithmic framework of RALLY by elaborating on the RMIX role assignment mechanism and context-migration algorithm. Section \ref{sec:results} presents experimental results and discussion. 
Finally, Section \ref{sec:conclusion} summarizes this work with possible research directions. 
\section{System Model and Problem Formulation}
\label{sec:model}
\subsection{System Model}
Beforehand, commonly used variables are summarized in Table \ref{tab:symbols}.
\begin{table}[t]
    \centering
    \caption{Mainly used notations in this paper.}
    \label{tab:symbols}
    \begin{tabular}{cl}
    \toprule
    \textbf{Symbol}  & \textbf{Description} \\
    \midrule
    $n$            & Number of UAVs \\
    $\mathcal{N}$  & Set of UAVs \\
    $\mathcal{T}$  & Set of target regions \\
    $\bm{p}_\text{tr}$, $\kappa_\text{tr}^t$ & Target positions/urgency \\
    $\bm{p}_e^t$, $\bm{v}_e^t$ & Enemy's position/velocity\\
    $\bm{p}_i^t$, $\bm{v}_i^t$ & Agent $i$'s position/velocity \\
    $\bm{s}^t$  & Global state \\
    \hline
    $\mathcal{N}_i^t$  & Set of agent $i$'s neighbors at time $t$ \\
    $\bm{o}_i^t$  & Local Observation of agent $i$\\
    $\bm{z}_i^t$  & Local information of agent $i$ \\
    ${k}_{i}^{t}$ & Role assignment of agent $i$\\
    $\mathcal{G}$ & Set of candidate target goals\\
    ${g'}_{i}^{t}$ & Initial goal intention of agent $i$\\
    $g_{i}^{t}$ & Final consensus goal of agent $i$\\
    \hline
    $\text{RL}_\text{HI}$ & Credit-assignment role selection   \\
    $\text{LLM}_\text{HC}$ & Two-stage LLM-based consensus inference \\
    $\text{LLM}_{\text{init}}$ & Initial intent generation policy   \\
    $\text{LLM}_{\text{cons}}$ & Consensus-refinement policy         \\
    $\text{X}_\text{task}$ & Task instruction prompt            \\
    $\text{Y}_\text{init}$ & Initial LLM prompt                 \\
    $\text{Y}_\text{cons}$ & Consensus-refinement LLM prompt    \\
    $\text{M}_\text{CoT}$ & Chain-of-Thought guidance prompts   \\
    \bottomrule
    \end{tabular}
    \end{table}

Hereinafter, we consider a DS-CEFC task, wherein $n$ communication-limited UAVs in set $\mathcal{N}$ cooperatively cover multiple target regions $\mathcal{T}$, with 2D positions $\bm{p}^t_\text{tr}$ and urgency levels $\kappa_\text{tr}^{t}$ that decreases over coverage time, $\text{tr} \in \mathcal{T}$, alongside the blocking from static obstacles and chasing from PPO-based \cite{PPO} adversarial ``enemy''. Therefore, during navigation, agents shall dynamically change the target region from a set of candidate target goals $\mathcal{G}^t \subset\mathcal{T}$ to balance evasion and coverage efficiency. Due to the partial observability and communication limitation within the observation range $\delta_\text{obs}$, as illustrated in Fig. \ref{DES-EFC}, each agent $i$ obtains an environmental observation as
\begin{equation} \label{eq:observation} \bm{o}_i^t := \{(\bm{p}_i^t, \bm{v}_i^t),\; (\bm{p}_e^t, \bm{v}_e^t), (\bm{p}_\text{tr}^t,\kappa_\text{tr}^{t})_{\text{tr}\in\mathcal{T}}\}, \end{equation} 
where $\bm{p}_i^t=[p_{x_i}^t,p_{y_i}^t]$ and $\bm{v}_i^t=[v_{x_i}^t,v_{y_i}^t]$ are the 2D position and velocity of agent $i$, respectively, while $\bm{p}_e^t=[p_{x_e}^t,p_{y_e}^t]$, $\bm{v}_e^t=[\bm{v}_{x_e}^t,\bm{v}_{y_e}^t]$ are those of enemy agent. Based on its local observation $\bm{o}_i^t$, agent $i$ first generates a target intention ${g}_i^t \in \mathcal{T}$ and role assignment $k_i^t \in \mathcal{K} =\{\texttt{Commander}, \texttt{Coordinator}, \texttt{Executor}\}$ (i.e., $\bm{a}_i^t = ({g}_i^t, k_i^t)$), and then exchanges these proposals and observation (denoted as $\bm{M}_\text{i}^t=(\bm{a}_i^t, \bm{o_i^t})$) within its neighborhood set $ \mathcal{N}_i^t \in \mathcal{N}$, ultimately forming a collective \emph{consensus} $g_i^t \in \mathcal{T}$ through semantic negotiation. Notably, different roles are associated with distinct decision-making strategies: \texttt{Commander} focuses on maximizing individual rewards through independent decision-making, tending to select points yielding the highest personal return;  \texttt{Coordinator} balances team and individual gains, giving priority to the \texttt{Commander}'s choices whenever necessary; \texttt{Executor} primarily adheres to the \texttt{Coordinator}'s guidance, reverting to its own strategy if necessary for task reliability. As discussed later, the heterogeneous role assignment contributes to the consensus inference of large-scale swarms. 
\begin{figure}[!t]
\centering
\includegraphics[width=1.0\columnwidth]{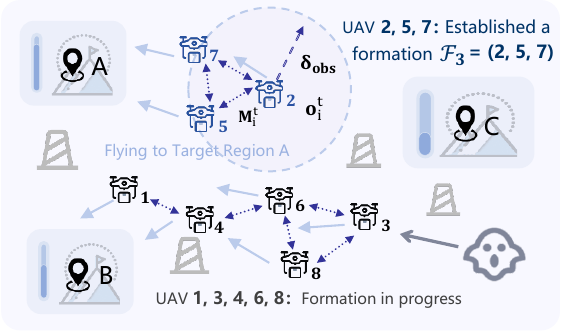}
\caption{A typical DS-CEFC scenario: Agents 2, 5, and 7 communicate directly within their neighborhood and then form a formation to cover target region A, while agents 1, 3, 4, 6, and 8 use indirect, multi-hop communication via reachable neighbors to form a formation and move toward target region B.}
\label{DES-EFC}
\end{figure}
\begin{figure}[!t]
\centering
\includegraphics[width=1.0\columnwidth]{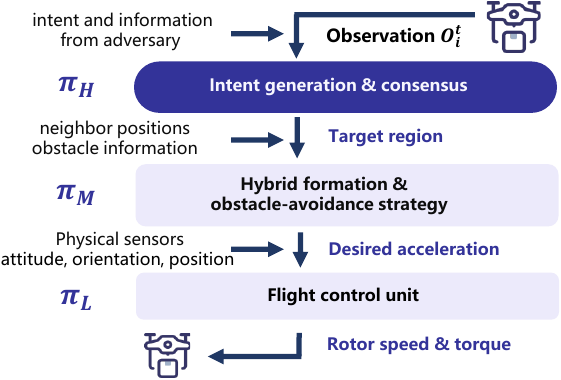}
\caption{The multi-layer policy architecture in distributed cooperative systems for UAVs.}
\label{ContorlSystem}
\end{figure}

As shown in Fig. \ref{ContorlSystem}, we further adopt a hierarchical policy structure. The mid-layer policy $\bm{\pi}_M$ guides the obstacle avoidance and formation $\mathcal{F}_c$, consisting of different numbers $c \in [C_{\min}, C_{\max}]$ via MARL-based navigation, while the low-layer standard PID \cite{meier2015px4} $\bm{\pi}_L$ steers the flight control of UAV dynamics. In this paper, we assume the availability of $\bm{\pi}_M$ and $\bm{\pi}_L$ \cite{CIHRL2025}, and focus on the learning of the high-level policy $\bm{\pi}_H$ only, which is responsible for the inferred consensus $\mathcal{G}^t:= \{g_i^t, \forall i \in \mathcal{N}\}$.  

\begin{figure*}[!t]
    \centering
    \includegraphics[width=1.0\textwidth]{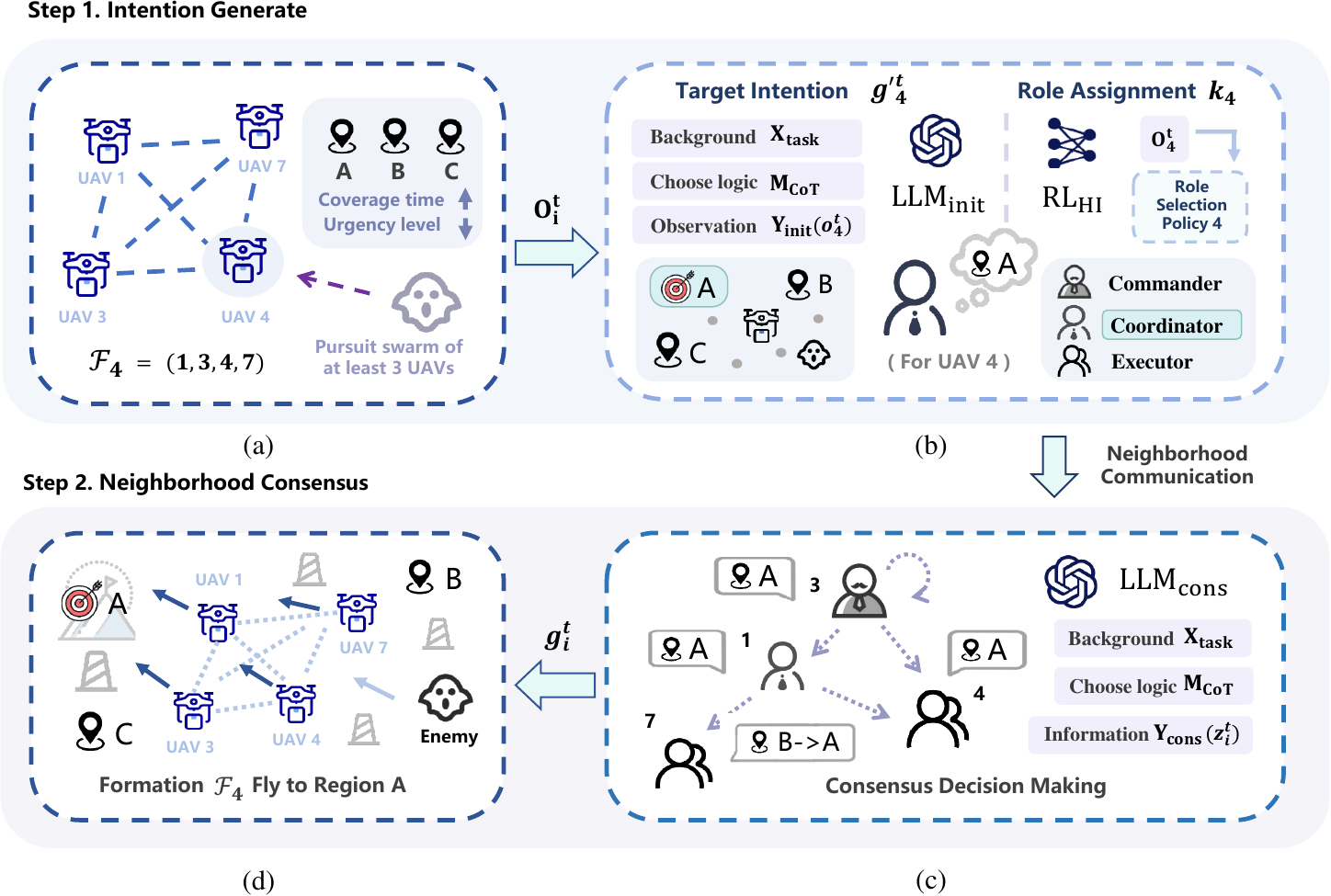}
    \caption{Flowchart of the RALLY algorithm. 
    (a) Taking formation $\mathcal{F}_{4}$ as an example, UAVs 1, 3, 4, and 7 adaptively form the team. 
    (b) Based on local observations $\bm{o}_i^t$, each agent selects target intention ${g’}_i^t$ and role $k_i$.
    Based on background prompt descriptions $\text{X}_\text{task}$, reasoning logic $\text{M}_\text{CoT}$, and local observation prompts $\text{Y}_\text{init}(\bm{o}_i^t)$, UAVs generate target intentions via policy $\text{LLM}_\text{init}$, and then determine roles $k_i$ through a role selection policy $\text{RL}_\text{HI}$.
    (c) After neighborhood communications, $\text{LLM}_\text{cons}$ refines the goal to reach a final, converged consensus decision $g_i^t$ (e.g., Region A), incorporating individual role preferences. 
    (d) The resulting formation $\mathcal{F}_{4}$ then collectively navigates toward Region A.}
    \label{RALLY}
\end{figure*}
\subsection{Problem Formulation}
We define the DS-CEFC task as Decentralized Partially Observable Markov Decision Process (Dec-POMDP) extended to a multi-agent setting with heterogeneous roles and hierarchical semantics. The joint system state at time $t$ is
\begin{equation}
\bm{s}^t := \{(\bm{p}_i^t,\bm{v}_i^t)_{i\in\mathcal{N}}, (\bm{p}_{e}^t,\bm{v}_{e}^t),\left(\bm{p}_\text{tr}^{t},\kappa_\text{tr}^{t})_{\text{tr}\in\mathcal{T}}\right \}. 
\end{equation} 
The state evolves stochastically according to $\bm{s}^{t+1} \sim \mathscr{P}(\bm{s}^{t+1} \mid \bm{s}^t, \bm{a}^t)$,
where the joint action of all $N:=|\mathcal{N}|$ agents is denoted as $\bm{a}^t:=(\bm{a}_1^t,\dots,\bm{a}_N^t)$. By contrast, each agent $i$ only has a local observation $\bm{o}_i^{t}$ given by Eq. \eqref{eq:observation}. 
Due to its long-lasting impact, the reward $R_t$ at time $t$ is formulated as a weighted sum of formation reward $R_{\text{f}}^t$, navigation reward $R_{\text{n}}^{t}$, mission completion $R_{\text{tc}}^{t}$, interference penalty $R_{\text{e}}^{t}$, and collision penalty $R_{\text{c}}^{t}$. In other words, 
\begin{equation}
R^{t}(\bm{s}^t , \bm{a}^t ) = \omega_{\text{f}} R_{\text{f}}^{t} + \omega_{\text{n}} R_{\text{n}}^{t} + \omega_{\text{tc}} R_{\text{tc}}^{t} - \omega_{\text{e}} R_{\text{e}}^{t} - \omega_{\text{c}} R_{\text{c}}^{t},
\end{equation}
with nonnegative scalars $R_{\text{f}}^t,R_{\text{n}}^{t},R_{\text{tc}}^{t},R_{\text{e}}^{t},R_{\text{c}}^{t}\ge0$ and respective weights $\omega_{\text{f}}$, $\omega_{\text{n}}$, $\omega_{\text{tc}}$, $\omega_{\text{e}}$, $\omega_{\text{c}}$. For each reward component, we adopt the consistent definition as in Ref. \cite{CIHRL2025}.
Then the total expected return under a joint policy $\bm\pi_H:=(\pi_1,\dots,\pi_N)$ is given by
\begin{equation}
    J(\bm\pi_H) = \mathbb{E}\left[\sum\nolimits_{t=0}^{\infty} \gamma^t\,R^t\right],
\end{equation}
where $0<\gamma\le1$ is a discount factor. Contingent on the readiness of $\bm{\pi}_M$ and $\bm{\pi}_L$ \cite{CIHRL2025}, the objective is to find decentralized policies $\bm\pi_H$ that maximize $J$ under a partially observable environment.

\section{Role-Adaptive LLM-driven Yoked Navigation}
\label{sec:rally}
In this section, we focus on the hierarchical design of RALLY for yielding the high-layer policy $\bm\pi_H$. Notably, RALLY primarily includes an LLM-based personalized consensus generation framework $\text{LLM}_\text{HC}$ and a credit-based role assignment mechanism $\text{RL}_\text{HI}$. 
\subsection{Two Stage LLM-Based Consensus Inference}
The high-level consensus generation strategy $\text{LLM}_\text{HC}$ employs a two-stage structured prompting approach—local intention generation $\text{LLM}_\text{init}$ followed by neighborhood consensus refinement $\text{LLM}_\text{cons}$—to realize an LLM-driven personalized semantic decision mechanism. 
By integrating role definitions, threat analysis, and formation coverage requirements into natural language prompts, numerical observations are mapped into interpretable intentions. This design deeply couples the LLM's semantic reasoning with various role logics. 
Notably, the LLM reasoning currently does not rely on historical memory and solely on the given context.

\begin{figure*}[t!]
\centering
\includegraphics[width=1.0\textwidth]{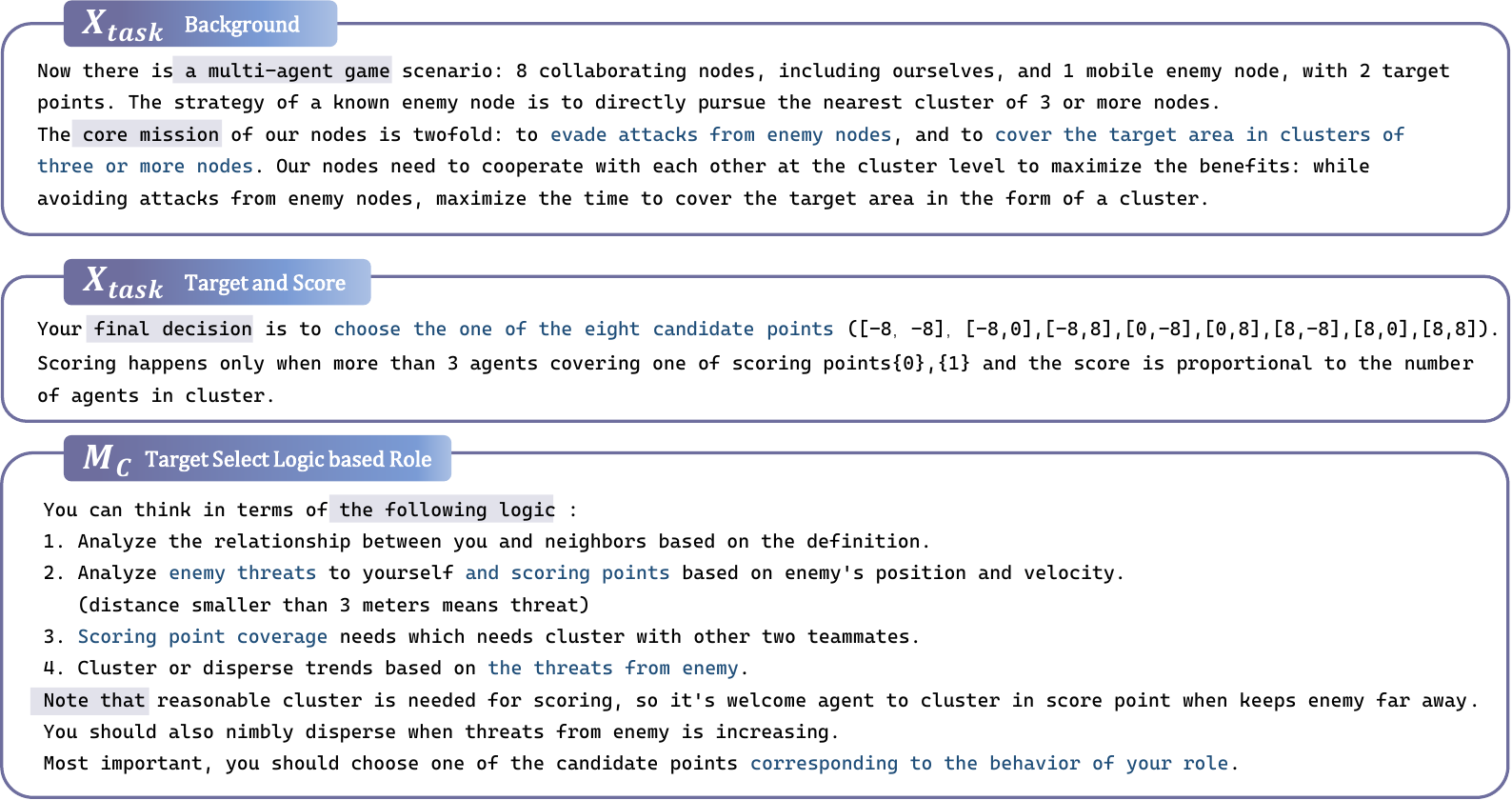}
\caption{Example prompts used as inputs to the LLM.}
\label{Prompt}
\end{figure*}

Firstly, each agent $i$ uses its current local observation $\bm{o}_i^t$ to generate \textit{initial intent} through the LLM. 
Accordingly, we design a task-specific instruction $\text{X}_\text{task}$ that outlines agents' mission requirements and task conditions, and craft the prompt ${\text{Y}_{\text{init}}}$ to capture the observation $\bm{o}_{i}^{t}$ of agent $i$. The LLM's natural-language output is then parsed into numeric goals ${g'}_i^t$ as
\begin{equation}
    \label{eq:intent}
    {g’}_i^t=\text{LLM}_\text{init}(\text{X}_\text{task},{\text{Y}_\text{init}}(\bm{o}_{i}^{t}),\text{M}_\text{CoT}).
\end{equation}
Next, the Role Selection Policy $\text{RL}_\text{HI}$ takes the current observation  $\bm{o}_{i}^{t}$ as input to optimize the role selection for each agent. The final selected role ${k}_i^t$ is given as
\begin{equation}
    \label{eq:rl}
    {k}_i^t=\text{RL}_\text{HI}(\bm{o}_{i}^{t} ; {g'}_i^t),
\end{equation}
which determines the agent's action strategy based on its role.

After communication with neighbors $j$ within a range $\delta_{\text{obs}}$ (i.e., $\|\bm{p}_i^t-\bm{p}_j^t\| < \delta_\text{obs}, \forall j \in \mathcal{N}_{i}^t$), the local information available to agent $i$ at time $t$ can be denoted as:

\begin{equation}
\bm{z}_i^t = \left(\bm{o}_i^t, \left\{(\bm{p}_j^t, \bm{v}_j^t, {g'}_j^t, k_j^t) \right\}_{j \in \mathcal{N}_{i}^t} \right).
\end{equation}

On this basis, each agent $i$ constructs a new prompt ${\text{Y}_{\text{cons}}}$ describing its intents, $\text{RL}_\text{HI}$-produced role $k_i^t$, neighbors' roles, as well as environmental constraints. We further incorporate task-driven Chain-of-Thought (CoT) prompts $\text{M}_\text{CoT}$, such as  ``\textit{Cluster or disperse based on the threats from enemy}" and ``\textit{needs cluster with other two teammates}", to strengthen the LLM's reasoning and minimize hallucinations. 
The LLM is instructed to output a \textit{refined consensus} like ``\textit{I recommend going to target $[x,y]$}" which is converted into the final numeric target:
\begin{equation}
    \label{eq:consensus}
    g_i^t = \text{LLM}_{\text{cons}}(\text{X}_\text{task},\text{M}_\text{CoT},{\text{Y}_\text{cons}}(\bm{z}_i^t )).
\end{equation}
This refinement embeds a balance between maximizing swarm utility and ensuring safe avoidance of enemy under adversarial and coverage requirements.

Example prompts for $\text{X}_\text{task}$ and $\text{M}_\text{CoT}$ are illustrated in Fig. \ref{Prompt}, while more prompt details, such as $\text{Y}_\text{init}$ and $\text{Y}_\text{cons}$, are provided in the Appendix.
Notably, as evidenced in Fig. \ref{llminput} of the Appendix, the design of prompts shows considerable sensitivity to the environmental changes and yields significantly different responses.

For occasional illegal LLM outputs, we implement hierarchical contingency strategies: \texttt{Commander} maintains initial intent; \texttt{Coordinator} defers to valid \texttt{Commander} (else self-reverts); \texttt{Executor} follows any available \texttt{Commander}/\texttt{Coordinator}. This design eliminates performance dips, and stabilizes system bounds by preventing collective deviations.

We compare the two-stage coordination policy  $\bm{\pi}^{\text{two}}_H=(\text{LLM}_{\text{cons}}\circ \text{LLM}_{\text{init}}) \rightarrow g_i^t$ against a one-stage baseline policy $\bm{\pi}^{\text{one}}_H=\widetilde{\text{LLM}}_\text{HC} \rightarrow \tilde{g}_i^t$, where the symbol $\circ$ marks a strategy synergy operation.
Beforehand, we introduce the following definition and assumption. 
\begin{definition}[Joint Policies Definition]

For agent $i$, the action-value function at time $t$ is: 
\begin{equation}
\label{def:value_functions}
\begin{aligned}
Q_i(\bm{o}_i^t, \bm{a}_i^t) & = \mathbb{E}\left[\sum_{t'=t}^{\infty} \gamma^{t'-t}\,R(\bm{s}^{t'},\bm{a}^{t'}) \;\middle|\; \bm{o}_i^t,\bm{a}_i^t \right] \\& = \mathbb{E}\left[\sum_{t'=t}^{\infty} \gamma^{t'-t}\,R(\bm{s}^{t'},\bm{a}^{t'}) \;\middle|\; \bm{o}_i^t,k_i^t,g_i^t \right].
\end{aligned}
\end{equation}
For global state $\bm{s}^t$ and all agents' goal decisions $\bm{g}^t:=(g_1^t,\cdots,g_N^t)$, the global value function $Q_\text{tot}$ is defined by a mixing network $f$, which will be detailed in Section \uppercase\expandafter{\romannumeral3}.\ref{sec_rmix}:
\begin{equation}
     Q_{\text{tot}}(\bm{s}^t,\bm{a}) = f(Q_1, Q_2, \dots, Q_N;\;\bm{s}^t,\bm{a}^t\bigr).
\end{equation}
\end{definition}
\begin{assumption}[Monotonic Value Factorization]
    \label{as:monotonic}
    Under a monotonic value factorization \cite{QMIX}, 
    every weight in $f$ is constrained to be non-negative, which guarantees the monotonicity property:
\begin{equation} 
\label{def:monotonic}
\frac{\partial Q_{\text{tot}}(\bm{s}^t, \bm{g})}{\partial Q_i} \geq 0, \quad \forall i.
\end{equation} 
\end{assumption}

    \begin{assumption}[Performance Improvement of Extra Contextual Reasoning]
        \label{as:two_stage_enhancement}
        Similar to \cite{GPT3,cotpe,llmthink}, the extra contextual reasoning helps guide the decision-making process toward more robust and effective outcomes, which leads to higher-quality $Q$-values.
    \end{assumption}

\begin{theorem}[Two-Stage Superiority]
    \label{thm:superiority}
Under Assumption \ref{as:monotonic} and Assumption \ref{as:two_stage_enhancement}, the two-stage policy achieves strictly higher expected return:
\begin{equation}
    J(\bm{\pi}^{\text{two}}) > J(\bm{\pi}^{\text{one}}),
\end{equation}
provided there exists at least one reachable context where $\text{LLM}_{\text{cons}}$ improves upon $\widetilde{\text{LLM}}_\text{HC}$.
\end{theorem}
\begin{proof}[Proof]
Under Assumption \ref{as:two_stage_enhancement}, for any $(\bm{o}_i^t, k_i^t)$, the two-stage consensus step selects 
\begin{equation}
Q_i(\bm{o}_i^t, k_i^t, g_i^t) \geq Q_i(\bm{o}_i^t, k_i^t, {g'}_i^t), \forall {g'}_i^t
\end{equation}
by construction of $\text{LLM}_{\text{cons}}$. 
In contrast, the one-stage policy outputs $\tilde g_i = \widetilde{\text{LLM}}_\text{HC}(\bm{o}_i^t,k_i^t)$ without refinement. In other words, Assumption \ref{as:two_stage_enhancement} implies
\begin{equation} 
Q_i(\bm{o}_i^t, k_i^t, g_i^t) \geq Q_i(\bm{o}_i^t, k_i^t, \tilde{g}_i^t). \end{equation}
Under Assumption \ref{as:monotonic}, 
increasing any individual agent' value $Q_i$ cannot decrease the total value,
which implies: 
\begin{equation} Q_{\text{tot}}(\bm{s}^t, \bm{g}^*) \geq Q_{\text{tot}}(\bm{s}^t, \tilde{\bm{g}}), 
\end{equation} 
with strict inequality if any agent achieves improvement.
\end{proof}
Theorem \ref{thm:superiority} shows that under Assumptions \ref{as:monotonic} and \ref{as:two_stage_enhancement}, the two-stage policy $\bm{\pi}^{\text{two}}_H=(\text{LLM}_{\text{cons}}\circ \text{LLM}_{\text{init}})$ achieves strictly higher expected cumulative reward than the one-stage baseline $\bm{\pi}^{\text{one}}_H=\widetilde{\text{LLM}}_\text{HC}$. In other words, while it may not guarantee the optimal goal decision, the two-stage consensus process provides a more comprehensive evaluation than the initial decision, potentially leading to an improved target decision. Next, we introduce role-based credit-assignment mechanism to ensure that Assumption \ref{as:monotonic} always holds.

\subsection{Credit-Assignment Mechanism Based on Role-value Mixing Network}
\label{sec_rmix}
\begin{figure}[!t]
\centering
\includegraphics[width=1.0\columnwidth]{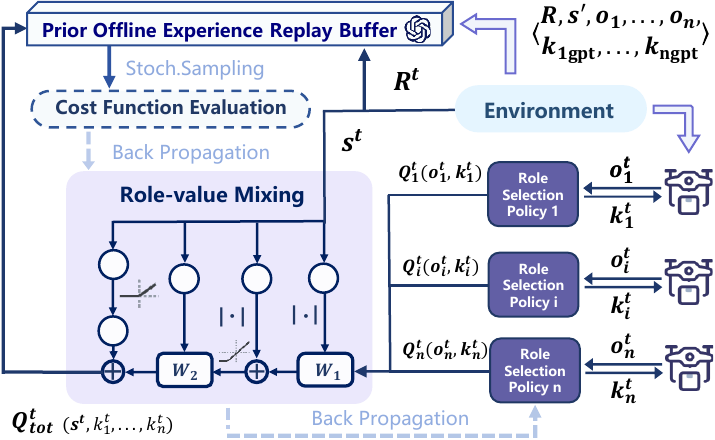}
\caption{The training of the RMIX-based role selection mechanism. Each UAV employs a fully distributed role selection policy, using local observation $o_i^t$ to selects role $k_i^t$, and computes corresponding $Q_i(\bm{o}_i^{t},k_i^t)$. These $Q_i$ are aggregated into the global action-value $\bm{Q}_\text{tot}$ via the RMIX network and stored in a prior offline experience replay buffer. Function evaluation cost is computed using the buffer data, followed by gradient backpropagation to update both the RMIX network and role selection policies.}
\label{RMIX}
\end{figure}
\begin{algorithm}[!ht]
\caption{Credit-Assignment Enhanced Intention Generation Strategy (RMIX)}
\label{alg:RMIX}
\begin{algorithmic}[1]
\REQUIRE Role candidate set $\mathcal{K} = \{\texttt{Commander}, \texttt{Coordinator}, \texttt{Executor}\}$, \\
\hspace{1.8em} RMIX parameters $\theta$, target network $\bar{\theta}$, \\
\hspace{1.8em} Discount factor $\gamma_{\text{rmix}}$, learning rate $\alpha$, batch size $b$, \\
\hspace{1.8em} Experience replay buffer $\mathcal{B}_R$.

\vspace{0.5em}
\STATE \textbf{Step 1: Offline Experience Collection via LLM}
\FOR{episode = 1 to $N_{\text{pre}}$}
    \STATE Initialize environment, observe $\{\bm{o}_i^t\}_{i=1}^n$;
    \FOR{each agent $i$}
        \STATE $k_{i_\text{GPT}} \leftarrow \bm{\pi}_{\text{GPT}}(\bm{o}_i^t)$; \hfill$\triangleright$ GPT assigns role
    \ENDFOR
    \STATE Execute actions, receive reward $R$, state $\bm{s}$;
    \STATE Store $\langle \bm{s}, \{\bm{o}_i, k_{i_\text{GPT}}\}, R \rangle$ into buffer $\mathcal{B}_R$;
\ENDFOR

\vspace{0.5em}
\STATE \textbf{Step 2: Online Cooperative RMIX Training}
\FOR{episode = 1 to $N_{\text{epoch}}$}
    \STATE Reset environment, observe $\{\bm{o}_i^0\}_{i=1}^N$;
    \FOR{$t = 0$ to $T - 1$}
        \FOR{each agent $i$}
            \STATE $Q_i(\bm{o}_i^t, k) \leftarrow \text{MLP}_\theta(\bm{o}_i^t)$;
            \STATE $k_i^t \leftarrow \arg\max_{k \in \mathcal{K}} Q_i(\bm{o}_i^t, k)$;
        \ENDFOR
        \STATE Execute roles $\{k_i^t\}$, receive $R^t$, $\bm{s}^{t+1}$, $\{\bm{o}_i^{t+1}\}$;
        \STATE Store $\langle \bm{s}^t, \{\bm{o}_i^t, k_i^t\}, R^t, \bm{s}^{t+1} \rangle$ into $\mathcal{B}_R$;
        
        \IF{$|\mathcal{B}_R| \geq b$}
            \STATE Sample batch $\{(\bm{s}^j, \{\bm{o}_i^j, k_i^j\}, R^j, \bm{s}^{j+1})\}_{j=1}^b$;
            \STATE Compute target values as follows:\\
            \(\vtop{\hbox{%
              $\displaystyle
                \begin{aligned}
                  y^j
                    &= R^j 
                     + \gamma_{\text{rmix}}\;\max_{\bm{k}'}\bar{Q}\bigl(\bm{s}^{j+1}, \bm{k}'; \bar{\theta}\bigr);
                \end{aligned}$
            }}\)

            \STATE Compute RMIX outputs $Q_{\text{tot}}$ according to Eq. \eqref{eq:rmix};
            
            \STATE Update: $\theta \leftarrow \theta - \alpha \nabla_\theta \sum_j (y^j - Q_{\text{tot}}^j)^2$;
            \STATE Periodically update target: $\bar{\theta} \leftarrow \tau \theta + (1 - \tau)\bar{\theta}$.
        \ENDIF
    \ENDFOR
\ENDFOR
\RETURN Trained RMIX network $\text{RL}_\text{HI}$ for decentralized role selection
\end{algorithmic}
\end{algorithm}
The underlying distributed nature of the DS-CEFC task makes a direct fine-tuning of LLM incompetent for reasonable multi-agent collaboration. 
As illustrated in Fig. \ref{RMIX}, each distributed agent $i$ feeds its current local observation $\bm{o}_i^t$, and computes the obtainable optimal role satisfying:
\begin{equation}
k_i^t=\underset{k_{i}^{t} \in\mathcal{K}}{\operatorname{argmax}}\ Q_{i}\left(\bm{o}_{i}^{t}, k_{i}^{t} ; {g'}_i^t\right).
\end{equation}
By Eq. \eqref{def:value_functions}, in dynamic multi-agent adversarial tasks, instantaneous reward $R$ cannot directly reflect the contribution from the choice of role $k_{i}^{t}$. To address this issue, we propose an RMIX-assisted credit-assignment mechanism to evaluate each agent's role choices. 
Also, we consider the significant inference latency of LLM, which results in higher costs for acquiring training samples.
In particular, with $8$ agents and $3$ roles each, the dimension of joint role space turns $3^8$, and when combined with sparse, adversarial rewards, this makes pure online exploration extremely difficult.
Therefore, we resort to an offline learning approach to improve sample efficiency.
Similar to Curriculum Learning (CL)\cite{Wang2020ACS}, we first exploit the zero-shot capability of LLM (e.g., GPT-4o\cite{Hurst2024GPT4oSC}) via API to obtain its role assignment strategy $\bm{\pi}_{\text{GPT}}$. 
To clarify, $\bm{\pi}_{\text{GPT}}$ refers to the online GPT4o-API version, rather than any locally fine-tuned large model.
By interacting with the environment under local observation $\bm{o}_i^t$, the mid-layer policy $\bm{\pi}_{M}$ and the physical flight controller $\bm{\pi}_{L}$, we record GPT's preliminary role suggestions ${k}_{i_\text{GPT}}$ and store the related transition quadruple $\langle {\bm o}_{i},{k}_{i_\text{GPT}},\bm{s}',R\rangle$ into a replay buffer $\mathcal{B}_R$, where $\bm{s}'$ denotes the transitioned state. This offline data collection runs in parallel with LLM consensus reasoning and seeds the replay buffer with sensible role assignments, enabling RMIX to assist the LLM in understanding role assignments and reducing the cold-start overhead from ineffective exploration.

Inspired by \cite{QMIX}, RMIX aggregates individual role-value estimates $Q_i(\bm{o}_i^{t},k_i^t)$ into a global value $Q_\text{tot}$ via a two-layer Multi-Layer Perceptron (MLP). 
Let $\bm{Q}^t=[Q_1(\bm{o}_1^{t},k_1^t),\ldots,Q_N(\bm{o}_N^{t},k_{N}^t)]$.
A small hypernetwork conditioned on the global state $\bm{s}^t$ produces nonnegative weights, by enforcing:
\begin{equation}
Q_\text{tot}^{t}=\text{ReLU}\left(\bm{W}_{2}^\top\big(\text{ReLU}(\bm{W}_{1} \bm{Q}^t+\bm{b}_{1})\big)+{b}_{2}\right),
\label{eq:rmix}
\end{equation}
where the two nonnegative weight vectors $\bm{W}_{1} \in \mathbb{R}^{E \times N}$ and $\bm{W}_{2}\in \mathbb{R}^{E \times 1}$, while  biases $\bm{b}_{1} \in \mathbb{R}^{{E}\times{1}}$ and $b_{2} \in \mathbb{R}$, with $E$ denoting the hidden layer dimension of RMIX. $\text{ReLU}(\cdot)$ indicates the ReLU nonlinear activation function. Eq. \eqref{eq:rmix} ensures the monotonic mapping between $Q_\text{tot}^t$ and $\bm{Q}^t$, thus satisfying the Assumption \ref{as:monotonic}.

Subsequently, we proceed to the standard CTDE online training. RMIX jointly learns from both offline and newly collected online samples to continuously explore and refine cooperative behaviors among agents. 
Generally, RMIX can be solved by a classical Stochastic Gradient Descent (SGD) approach, with the cost function of RMIX being formulated as:
\begin{equation}
L(\theta)=\sum_{i=1}^{|\mathcal{B}_R|}\left[\left( y_{i}^\text{tot}-Q_\text{tot}(\bm{s}^t,k^t; \theta,{g}^t)\right)^{2}\right], 
\end{equation}
where $\theta$ denotes the concatenation of $\bm{W}_{1}$, $\bm{W}_{2}$, $\bm{b}_{1}$ and $b_{2}$ while $|\mathcal{B}_R|$ denotes the number of samples sampled from the empirical memory $\mathcal{B}_R$. As an approximate estimation of the cumulative returns under the global state $\bm{s}^t$, $y^\text{tot}$ is given as:
\begin{equation}
y^\text{tot}=R^t+\gamma_\text{rmix} \max _{k^{\prime}} \bar{Q}\left(\bm{s}^{\prime},k^{\prime} ; \bar{\theta},{g}^t\right),
\end{equation}
where for the transitioned state $\bm{s}^{\prime}$ and the taken action $k^{\prime}$, $\bar{Q}\left(\bm{s}^{\prime},k^{\prime}; \bar{\theta}\right)$ denotes the target network parameterized by $\bar{\theta}$ and $\gamma_\text{rmix}$ is the discount factor.

Finally, Algorithm \ref{alg:RMIX} summarizes the procedure of RMIX.

\subsection{Context-Based LLM Fine-Tuning}
\begin{algorithm}[!t]
\caption{Context Transfer Enhanced Lightweight Large Model Consensus Inference Algorithm}
\label{alg:Lightweight}
\begin{algorithmic}[1]
\REQUIRE Task prompt $\text{X}_\text{task}$, prompt generator $\text{Y}(\bm{o}_i)$, \\
\hspace{0.2em} GPT‑4o API, target set $\mathcal{G}_a$, rewards $R$, \\
\hspace{0.2em} thresholds $\{\tau_r,L_{\min},L_{\max}\}$, weights $\{w_{1,g},\cdots, w_{4,g}\}$, \\
\hspace{0.2em} minimum samples $M$, batch size $B$, LoRA params $\psi$.

\STATE $\mathcal{B}_\text{LLM}\gets\emptyset$
\WHILE{$|\mathcal{B}_\text{LLM}|<M$}
  \STATE Observe $(\bm{o}_i)$;
  \STATE $(g_i^*,\mathscr{I}_i^*)\gets\text{GPT4o}(\text{X}_\text{task},\,\text{Y}(\bm{o}_i))$\;cf.\ Eq.~\eqref{human-mark};
  \STATE Append $(\bm{o}_i,k_i,g_i^*,\mathcal{I}_i^*,R)$ to $\mathcal{B}_\text{LLM}$;
\ENDWHILE

\STATE $\mathcal{B}_\text{fil}\!\gets\!\{x\in\mathcal{B}_\text{LLM}\mid \text{Check}(x)\ge1\}$\;cf.\ Eq.~\eqref{QualityEvaluation};

\FOR{epoch=1 \TO $N_{\mathrm{epoch}}$}
  \STATE Sample batch $\mathcal{B}\subset\mathcal{B}_\text{fil}$ of size $B$;
  \STATE Compute loss via Eq.~\eqref{LoraLoss};
  \STATE Update $\psi\gets\psi - \eta\nabla_\psi$;
\ENDFOR

\RETURN Fine‐tuned LLM$_\psi$
\end{algorithmic}
\end{algorithm}

Albeit the proficiency of LLMs like GPT-4o\cite{Hurst2024GPT4oSC}, a lightweight version is preferred by resource-constrained multi-agent system. Therefore, through self-generated instruction tuning \cite{Wang2022SelfInstructAL}, we introduce a capacity-migration augmentation, so as to improve a smaller model's reasoning ability.
Through a concerted effort, we successfully transfer task-understanding capabilities from a State-Of-The-Art (SOTA) foundation model to a local version, and subsequently compress the model to under 5GB of memory usage, thus potentially enabling distributed inference of the consensus reasoning module on onboard UAV GPUs\cite{pmlr-v235-park24e}.

Due to the limited dataset in DS-CEFC, we call an online GPT-4o model to generate samples as: 
\begin{equation}
\label{human-mark}
{X}_\text{task},\text{Y}(\bm{o}_{i}^t) \rightarrow g_{i}^{*},\mathscr{I}_i^{*},
\end{equation}
where $g_{i}^{*}$ denotes the manually labeled decision, and $\mathscr{I}_i^{*}$ constitutes additional consensus inference outputs, stored beyond Eq. (7), specifically for training purposes.
In particular, to guide the generation of desired output in Eq. \eqref{eq:consensus}, human-annotated instructions, similar to those in Fig. \ref{Prompt} and \ref{llmoutput}, are provided. On this basis, GPT-4o's output ${g}_{i}$ serves as RALLY's consensus decision for selecting the agent's target region. We then interact with the mid-layer policy $\bm{\pi}_{M}$ and the physical flight controller $\bm{\pi}_{L}$ to obtain the next local observation $\bm{o}_{i}^{t+1}$ and role $k_{i}^{t+1}$. This process repeats until acquiring a dataset $\mathcal{B}_\text{LLM}$ with sufficient inference samples and trajectory data $(\bm{o}_i,k_i,g_i^*,\mathcal{I}_i^*,R)$.

The raw datasets $\mathcal{B}_\text{LLM}$ may contain low-quality or invalid outputs. Therefore, we apply a filtering mechanism to retain a high-quality subset $\mathcal{B}_{\text{fil}}$. To weed out low-quality samples, we check whether the yielded target region $g_{i}^{*}$ belongs to the target set $ \mathcal{T}$, the length of the inference content $\mathscr{I}_i^{*}$ is illegal with anomalous symbols, and it conforms to the task constraints based on the $r_i$. Mathematically,
\begin{equation}
\label{QualityEvaluation}
\begin{aligned}
&\text{Check}\left(g_i^{*}, \mathscr{I}_{i}^{*}, R\right) \\
& = 
 \mathbb{I}\left(g_i^{*} \in \mathcal{G}_a\right) \cdot w_{1,g} + \mathbb{I}\left(\mathscr{I}_{i}^{*} \cap \Lambda = \emptyset\right) \cdot w_{2,g} \\
& + \mathbb{I}\left(L_{\min} \leq\left|\mathscr{I}_{i}^{*}\right| \leq L_{\max}\right) \cdot w_{3,g}  + \mathbb{I}\left(R \geq \tau_r\right) \cdot w_{4,g},
\end{aligned}
\end{equation}
where $\mathbb{I}\left(\cdot\right)$ is an indicator function that returns $1$ when the condition is established and $0$ otherwise; $w_{1,g}$, $w_{2,g}$, $w_{3,g}$, $w_{4,g}$ are the indicator importance weights; $\Lambda$ is a pre-defined set of anomalous characters; $L_{\min}$ and $L_{\max}$ denote the minimum and maximum inference result lengths, respectively.

Notably, the fine-tuning aims to minimize the mean squared error (MSE) between model outputs and GPT-4o reference samples, that is:
\begin{equation}\label{LoraLoss}
\begin{split}
\mathcal{L}_{\text{MSE}} = \frac{1}{\lvert\mathcal{B}_{\text{fil}}\rvert}
   \sum_{(g_{i}^{*},\mathscr{I}_i^{*}) \in \mathcal{B}_{\text{fil}}}
   \Bigl(
     \text{LLM}_\psi\bigl(\text{X}_{\text{task}},\,
       \text{Y}(o_i)\bigr)-
     (g_{i}^{*},\mathscr{I}_i^{*})
   \Bigr)^2,
\end{split}
\end{equation}
where $\psi$ denotes the low-rank adaptation parameters. To reduce computation and storage requirements, we use the LLaMA-Factory framework\cite{Zheng2024LlamaFactoryUE} with LoRA \cite{Hu2021LoRALA} to fine-tune a smaller, $\psi$-parameterized LLM (e.g., Qwen2.5 or Llama3.1).

Finally, the pseudocode is presented in Algorithm \ref{alg:Lightweight}. 
\section{Experimental Results and Discussions}
\label{sec:results}
\subsection{Experimental Settings}
\begin{table}[t]
    \centering
    \caption{Key Environmental and Reward Function Parameter Configurations}
    \label{tab:para_env}
    \begin{tabular}{l|c}
    \hline
    \textbf{Parameter} & \textbf{Setting} \\ \hline
    No. of UAVs $N$                  & $\{8,9,10,11\}$ \\
    Formation patterns set $\mathcal{C}$        & $\{3,4,5,6,7,8\}$ \\
    UAV velocity range (m/s)        & [$-1$, $1$]\\
    Adversary velocity range (m/s)  &[$-0.75$ , $0.75$]\\
    Observation distance $\delta_\text{obs}$  (m)        & $3$ \\
    Weights $\omega_{\text{f}}$, $\omega_{\text{n}}$, $\omega_{\text{tc}}$, $\omega_{\text{e}}$, $\omega_{\text{c}}$ of reward $R$     &$(15,4,10,100,100)$\\
    Decrease factor $\omega_d$                   & $0.003$\\
    Joint policy discount factor $\gamma $   & $0.92$    \\
    RMIX discount factor $\gamma_\text{rmix}$           & $0.95$  \\
    RMIX learning rate $\alpha$             & $1\times 10^{-5}$  \\
    RMIX hidden layer dimension $E$ & $128$\\
    Thresholds $\{\tau_r,L_{\min},L_{\max}\text{(token)\}}$    &$(-3,000,200,400)$ \\
    LoRa weights $w_{1,g}\cdots w_{4,g}$                   & $(0.45,0.25,0.2,0.1)$ \\
    Minimum No. of samples $M$                &  $12,000$\\
    \bottomrule
    \end{tabular}
\end{table}
\begin{figure}[!t]
\centering
\includegraphics[width=0.85\linewidth]{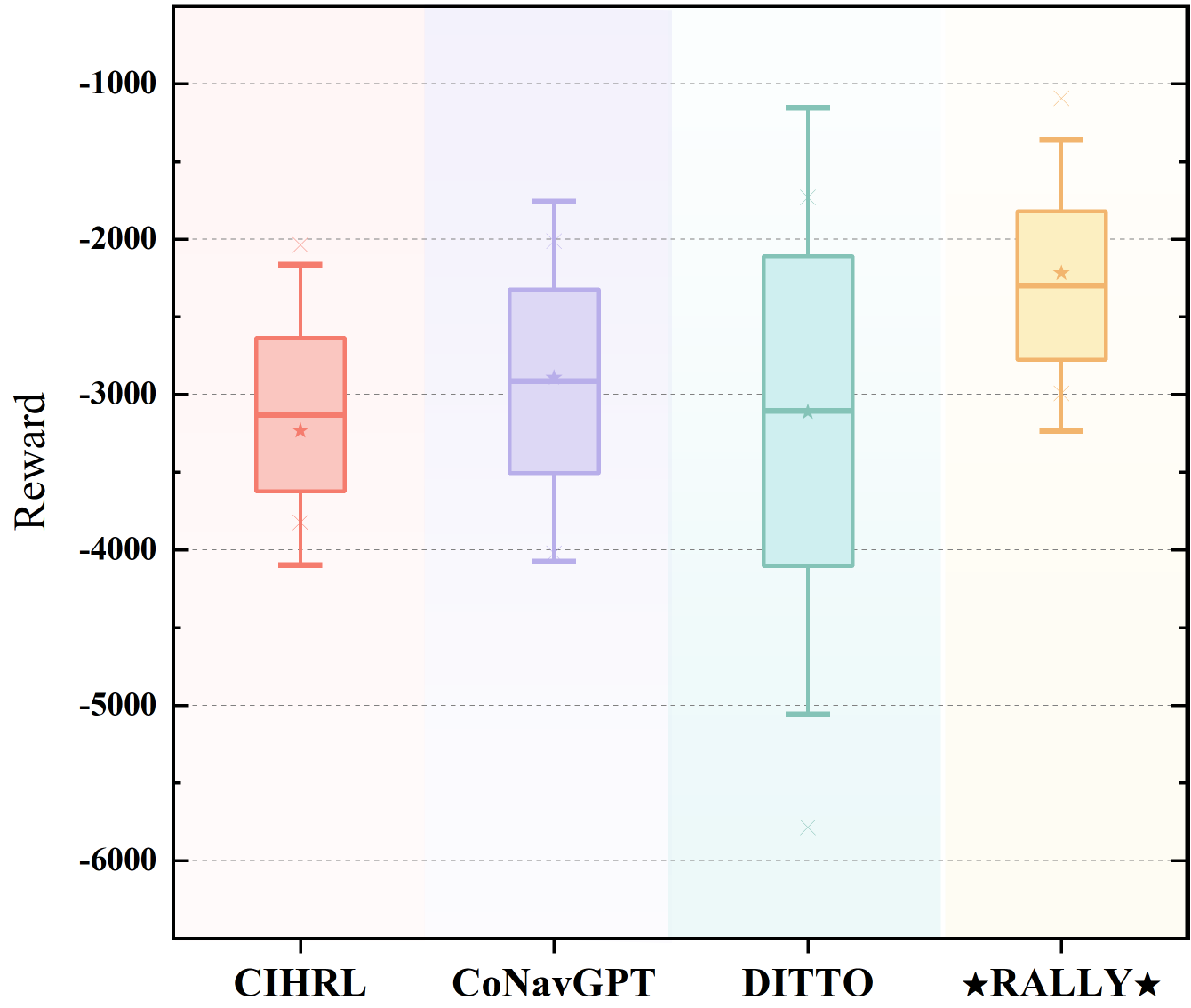} 
\caption{Performance comparison between RALLY and baselines.} 
\label{overall} 
\end{figure}
In this section, we evaluate the effectiveness and superiority of RALLY on DS-CEFC tasks in MPE\cite{lowe2017multi} and a fully distributed, high-fidelity SITL\cite{DangNguyen2019VisionBasedSF} simulation built on Gazebo–ROS–PX4. Concretely, each simulation runs for $1,000$ steps, wherein every $50$ steps, UAVs select their target regions according to $\pi_H$. The set of possible targets is $\mathcal{G}_a=\{(p_{x},p_{y})|p_{x},p_{y}\in \{-8, 0, 8\} \}$ with both coordinates initialized uniformly in $[-8, 8]$ m and re-sampled when a region's urgency reaches zero. The urgency level $\kappa_\text{tr}^{t}$ is initialized as $1$. Whenever there exists a formation of UAVs covering the region $\text{tr}\in \mathcal{T}$,  $\kappa_\text{tr}^{t}$ will decrease linearly with a scale $\omega_d$ until zero;
otherwise, it remains unchanged. 
The pursuer's speed is set to be twice that of the evader to ensure it has the necessary mobility to chase down and capture the evaders effectively.
Other key environment and mission parameters are listed in Table \ref{tab:para_env}. 

\begin{figure}[!t]
\begin{minipage}[t]{0.505\linewidth}
    \includegraphics[width=\linewidth]{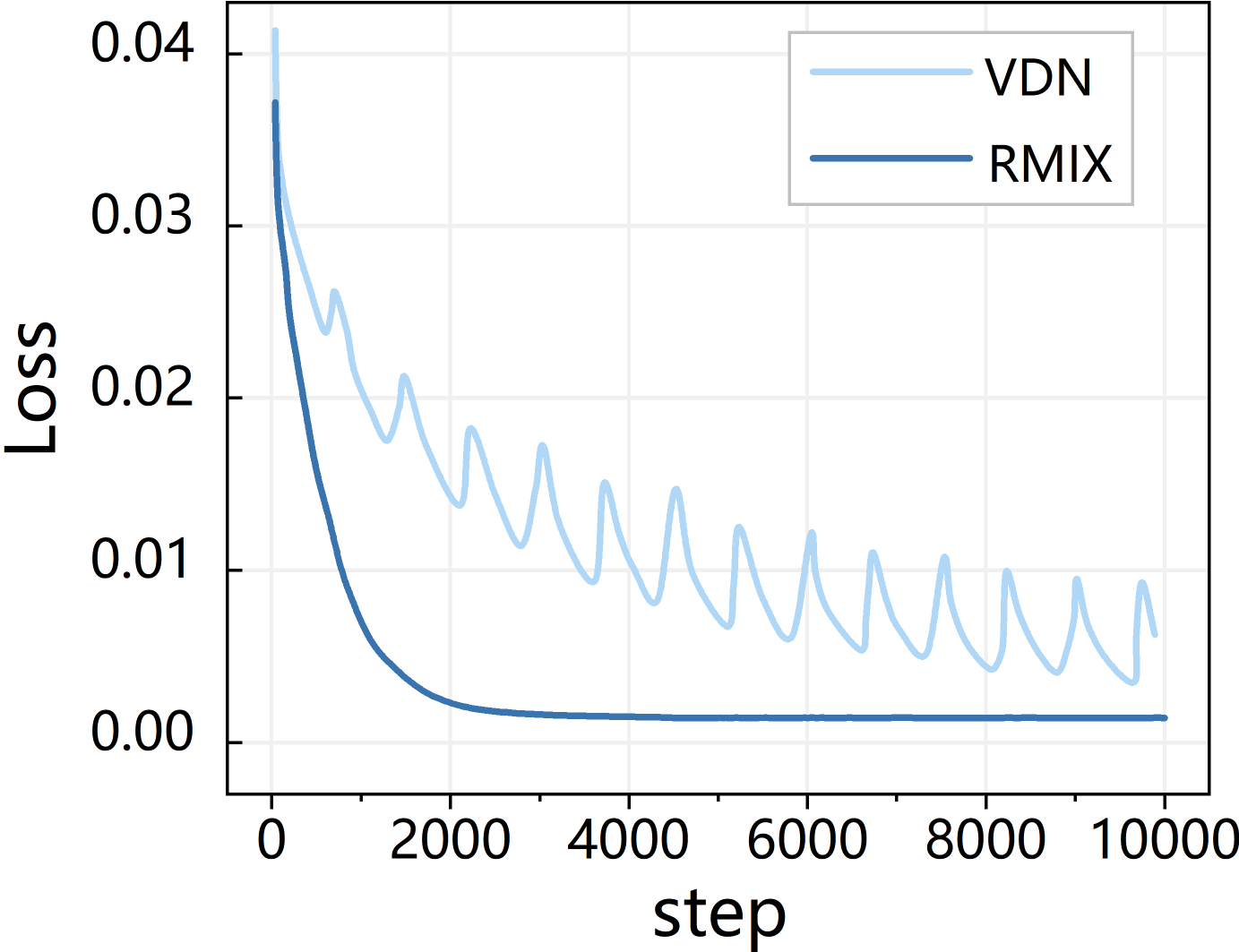}
    \caption{Training based on RMIX.} 
    \label{RmixLoss}
\end{minipage}%
    \hfill%
\begin{minipage}[t]{0.48\linewidth}
    \includegraphics[width=\linewidth]{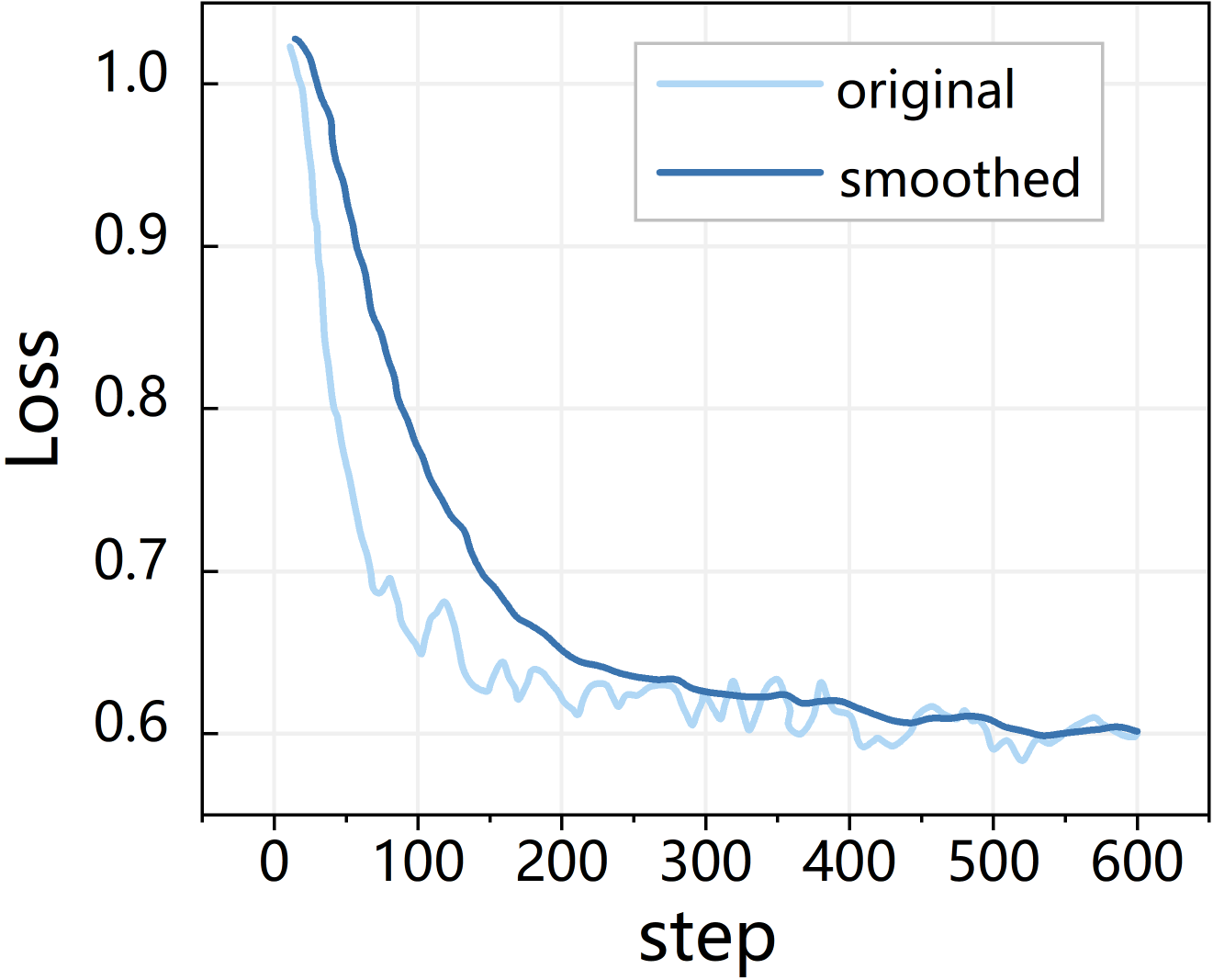}
    \caption{Fine-tuning of Qwen2.5-7B.}
    \label{LoraLosspic}
\end{minipage} 
\end{figure}
For RALLY, the RMIX network is implemented as a two-layer MLP with a hidden layer size of $16$. During training, we use a batch size of $256$, and update the target network with soft updates at a rate of $0.01$.
For the capacity-migration-augmentation strategy, we create $L=50$ manually labeled examples to serve as the few-shot corpus. Given the inference latency, we set the few-shot sample count to $\rho=1$. During the API-driven data collection, the GPT-4o model is leveraged to provide high-quality inferences. After simulation and filtering, we accumulate $|\mathcal{B}_\text{fil}|=8,231$ samples for fine-tuning the local model Qwen2.5-7B. 
The adversary uses a PPO policy\cite{PPO} with the same network structure as $\pi_M$, learning to chase the nearest cluster of at least three agents while avoiding obstacles.

\begin{figure}[t] 
\centering
\includegraphics[width=0.8\linewidth]{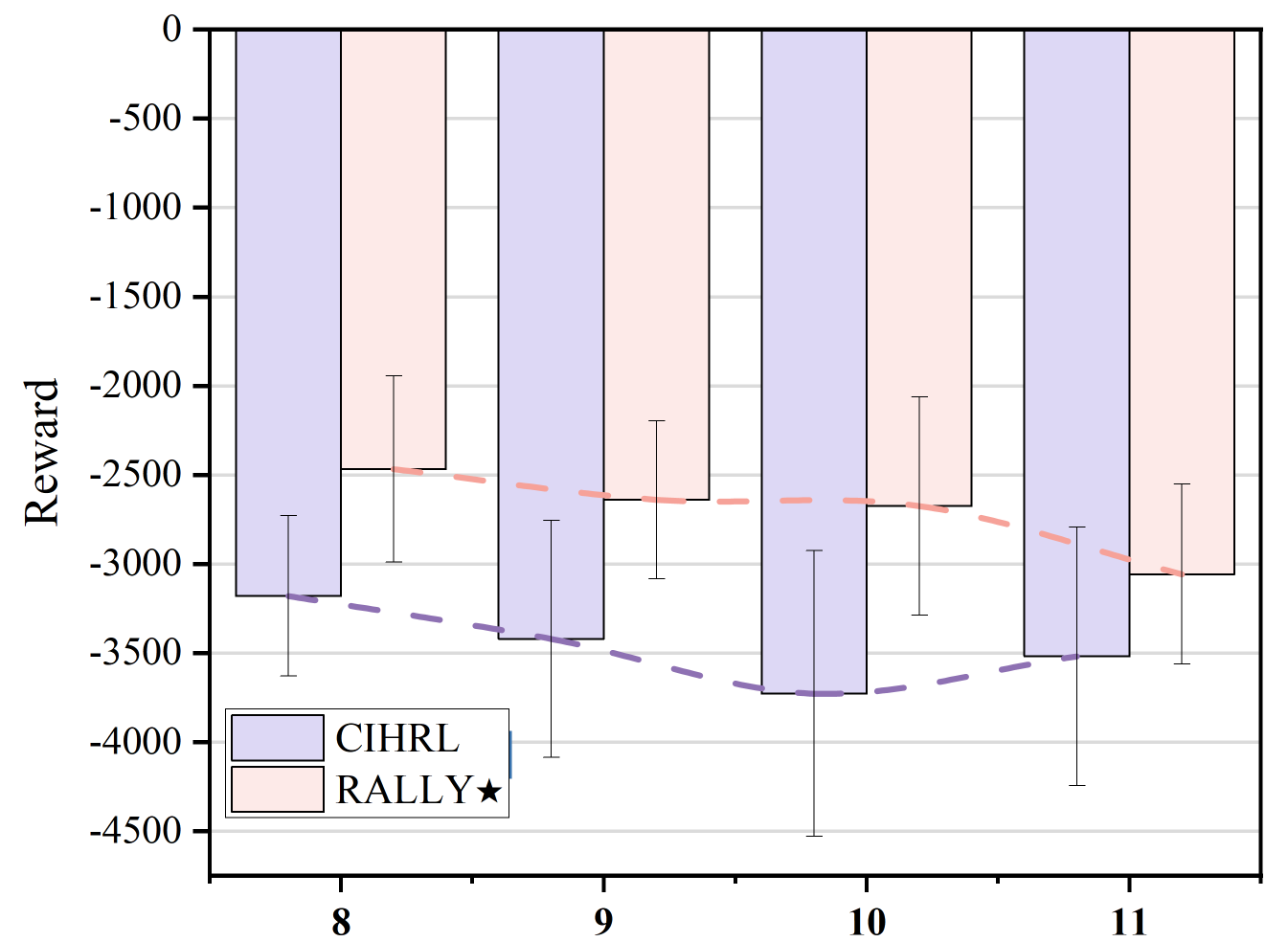}
\caption{Generalization of RALLY to varying numbers of agents.}
\label{AgentNum}
\end{figure}

We compare RALLY against three representative baselines, CIHRL\cite{CIHRL2025}, CoNavGPT\cite{yu2023co}, and DITTO\cite{ditto2024large}, to evaluate its effectiveness on the DS-CEFC task.  
CIHRL \cite{CIHRL2025}, which does not incorporate role assignment, incorporates multi-agent communication and belongs to the SOTA MARL approach for DS-CEFC.  
CoNavGPT \cite{yu2023co} employs an LLM as a global planner without any training process, achieving high success rates and efficiency on the navigation task. 
DITTO \cite{ditto2024large} achieves good collaboration based on LLM, demonstrating strong role-based heterogeneity.
All the non-API LLMs are deployed in a distributed fashion on $8$ NVIDIA GeForce RTX 4090 GPUs (24 GB each). 
\begin{figure}[t] 
\centering
\includegraphics[width=0.78\linewidth]{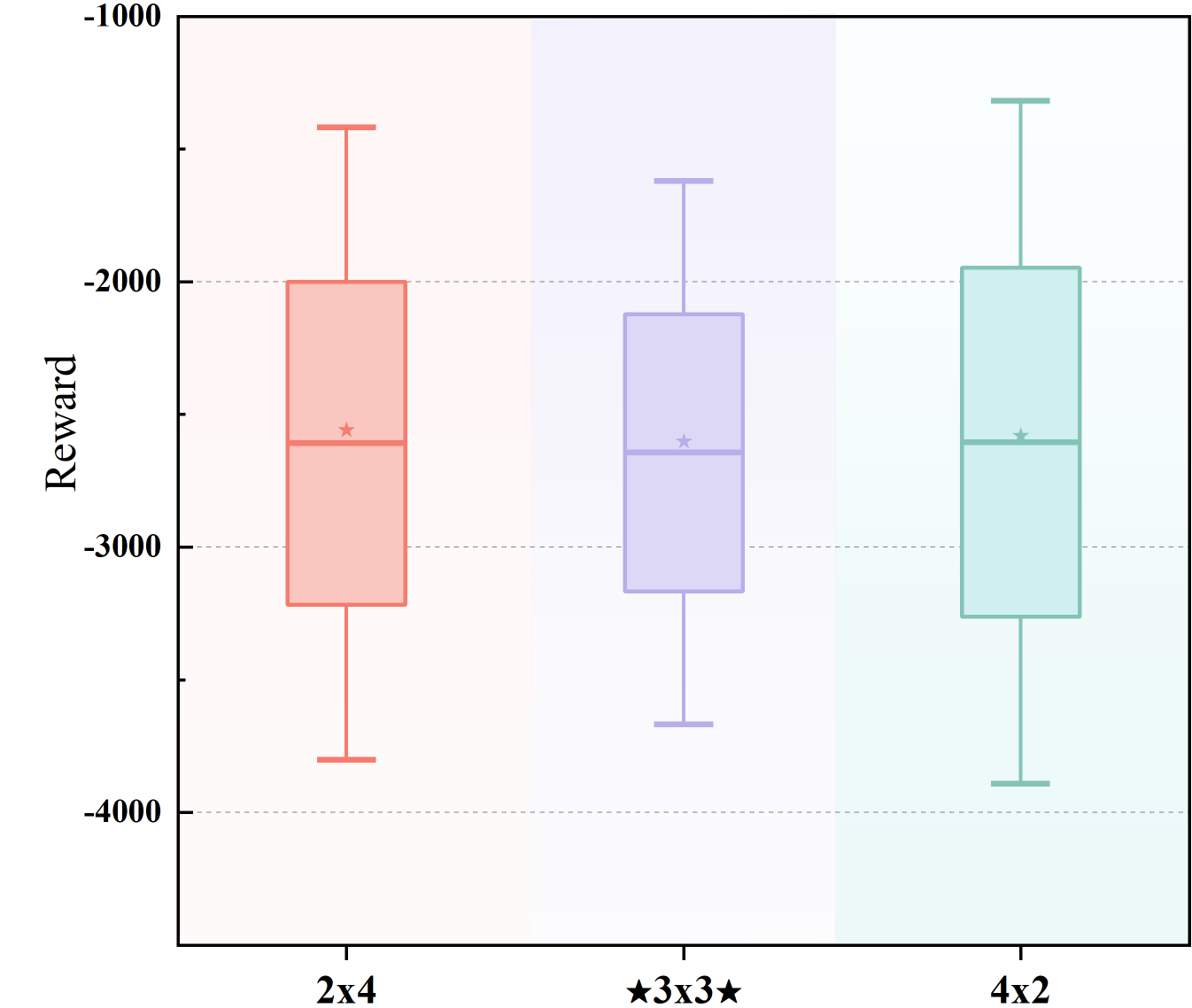}
\caption{Generalization of RALLY to varying target areas.}
\label{TargetArea}
\end{figure}
\begin{figure*}[!b]
\centering
\includegraphics[width=1.0\textwidth]{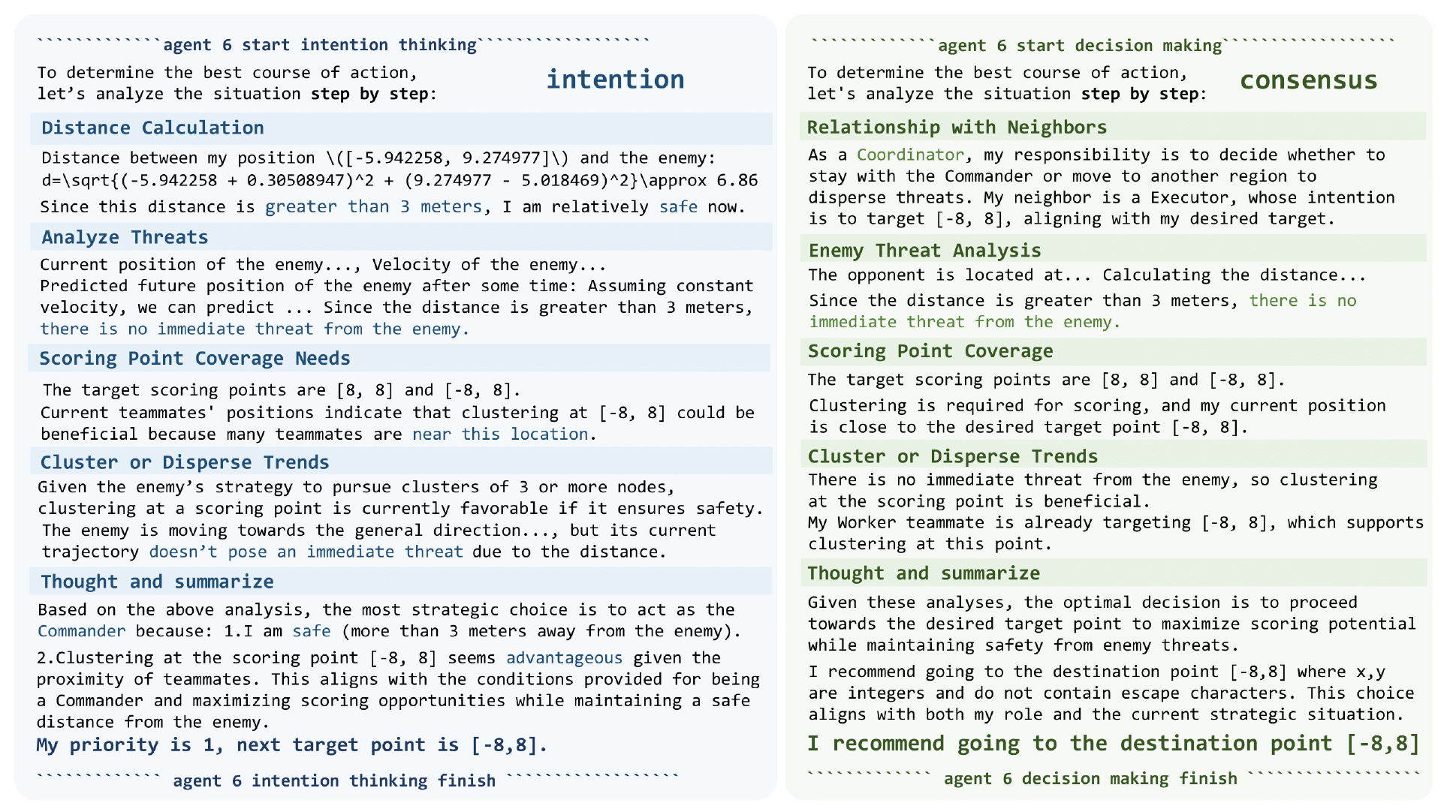}
\caption{A two-phase decision-making procedure executed by Agent \#6, where both stages are derived from the LLM outputs.} 
\label{llmoutput}
\end{figure*}
\subsection{Performance Comparison with Baselines}

\subsubsection{Overall Performance}

Fig. \ref{overall} presents the average rewards\footnote{Consistent with our previous works Ref. \cite{CIHRL2025}, due to the large negative values of interference penalty $R_{\text{e}}^{t}$, and collision penalty $R_{\text{c}}^{t}$, the reward defined in RALLY is generally negative as well.} over $30$ test episodes and demonstrates the impact of consensus mechanisms on task completion. Particularly, RALLY attains the highest mean reward and the narrowest variance distribution, indicating minimal variance, fastest convergence, and consistent task completion across various episodes. In contrast, CIHRL behaves conservatively, yielding stable but modest rewards.
CoNavGPT achieves higher average rewards than CIHRL, confirming the LLM's strong environmental understanding and decision-making, but it lacks online exploration and can get stuck in local optima.
DITTO slightly outperforms CIHRL by using LLM-based self-cognition for role and action selection, yet its greedy role choices and absence of reinforcement feedback lead to high variance and unstable consensus. 
These results demonstrate that RALLY's integration of RL–based environmental feedback and LLM-driven semantic decision making effectively guides the multi-UAV swarm to reach robust consensus and execute high-quality collaboration in complex dynamic environments.

Next, we study the convergence of the credit-based role assignment mechanism. Particularly, ``RMIX" uses the MLP-based fusion network for joint utility estimation detailed in Section \uppercase\expandafter{\romannumeral3}.\ref{sec_rmix}, while ``VDN" aggregates $Q_i$ via a simple weighted sum  $Q_{\text{tot}} = \sum_{i=1}^{n} w_i Q_i(o_i, k_i)$. Fig. \ref{RmixLoss} presents the corresponding results. 
Although both methods converge, RMIX converges faster and yields more accurate cumulative return estimates. Fig. \ref{LoraLosspic} shows the LoRA fine-tuning loss curve, while the validation starting from step $500$ indicates successful convergence. 
\subsubsection{Generalization}
As depicted in Fig. \ref{AgentNum}, we evaluate RALLY and CIHRL with changing swarm sizes from $8$, $9$, $10$ to $11$ on the DS-CEFC task. 
As the swarm grows beyond this training configuration, CIHRL's performance degrades substantially. In other words, without retraining for larger formations, the involved agents fall into ``habitual grouping'' patterns that repeatedly form the learned clusters, preventing effective coverage of additional targets and leading to significant score losses. In contrast, RALLY, which encodes the maximum permitted formation size into its prompt, dynamically forges a split consensus to avoid excessive clustering. More importantly, RALLY preserves high scoring ability even as the swarm size increases to $9$, $10$, and $11$. This clear contrast underscores RALLY's superior generalization to larger, unseen formations than CIHRL.

We further evaluate RALLY across three different target area configurations: the original $3\times3$ grid, a $2\times4$ grid, and a $4\times2$ grid. The target locations in each scenario are randomly generated, ensuring dynamic and diverse environments. As shown in Fig. \ref{TargetArea}, RALLY performs consistently well across all three scenarios, with no significant difference in reward performance, reinforcing its ability to adapt to varying environmental conditions.

\subsection{Performance Analysis of RALLY}

\subsubsection{Contribution of RMIX}
To illustrate how RMIX enhances LLM semantic decision making, we examine the differences between the initial intention and inferred consensus of Agent \#6, as shown in Fig. \ref{llmoutput}. In the initial LLM-only phase, Agent \#6 computes its distance to the adversary ($\approx 6.86$ m), and driven by ``safe aggregation'' and ``maximal scoring", greedily adopts the \texttt{Commander} role with target (-8,8). While this choice yields short-term points, it overlooks team coordination and can destabilize consensus under complex threats.   
Therefore, the role-selected policy overrides this role to \texttt{Coordinator}, shifting the agent's motive from individual dominance to supporting and scheduling. Receiving a neighbor's intent (also $(-8,8)$), the LLM then recommends $(-8,8)$ again. Moreover, the LLM also issues explicit role alignment that maintains both proximity to the \texttt{Executor} and collaboration with the \texttt{Commander}. 
This refined decision reuses the initial safety assessment and integrates multi-party communication via role mapping, yielding a more holistic trade-off between individual scoring and team consistency. By correcting the LLM's isolated ``\texttt{Commander}" bias, RALLY preserves LLM's semantic planning strengths while injecting MARL's distributed division of responsibility and information fusion. Consequently, the two-stage output achieves superior policy stability and collaborative performance in the dynamic DS-CEFC task.

\subsubsection{Impact of Role Number}
Fig. \ref{RoleNum} illustrates RALLY's reward distributions under four different role configurations, including single role (\texttt{Executor}), two roles (\texttt{Commander}–\texttt{Executor}), three roles (\texttt{Commander}–\texttt{Coordinator}–\texttt{Executor}), and four roles (\texttt{Commander}–\texttt{Coordinator}–\texttt{Executor}–\texttt{Decoy}), where the \texttt{Decoy} role is specifically designed to divert enemy attention. The single-role setup achieves the lowest mean reward ($\approx -3,000$), or severely limited performance due to the absence of task decomposition and limited exploration–coverage trade-off. Introducing a dual-role hierarchy (\texttt{Commander}–\texttt{Executor}) yields a slight increase in mean reward but greatly enlarged whiskers, indicating that overreliance on the \texttt{Commander}'s decisions amplifies fluctuation and undermines group synergy. In contrast, the three-role configuration combines high average performance with smaller variance, demonstrating that introducing a \texttt{Coordinator} role effectively mediates the semantic planning benefits from the LLM while enhancing consistency and robustness through MARL's exploration and credit-assignment. Unfortunately, adding a fourth \texttt{Decoy} role reduces average reward and inflates variance, suggesting that excessive role granularity raises coordination overhead and consensus costs, thereby detracting from overall effectiveness. 
\begin{figure}[t] 
\centering
\includegraphics[width=0.8\linewidth]{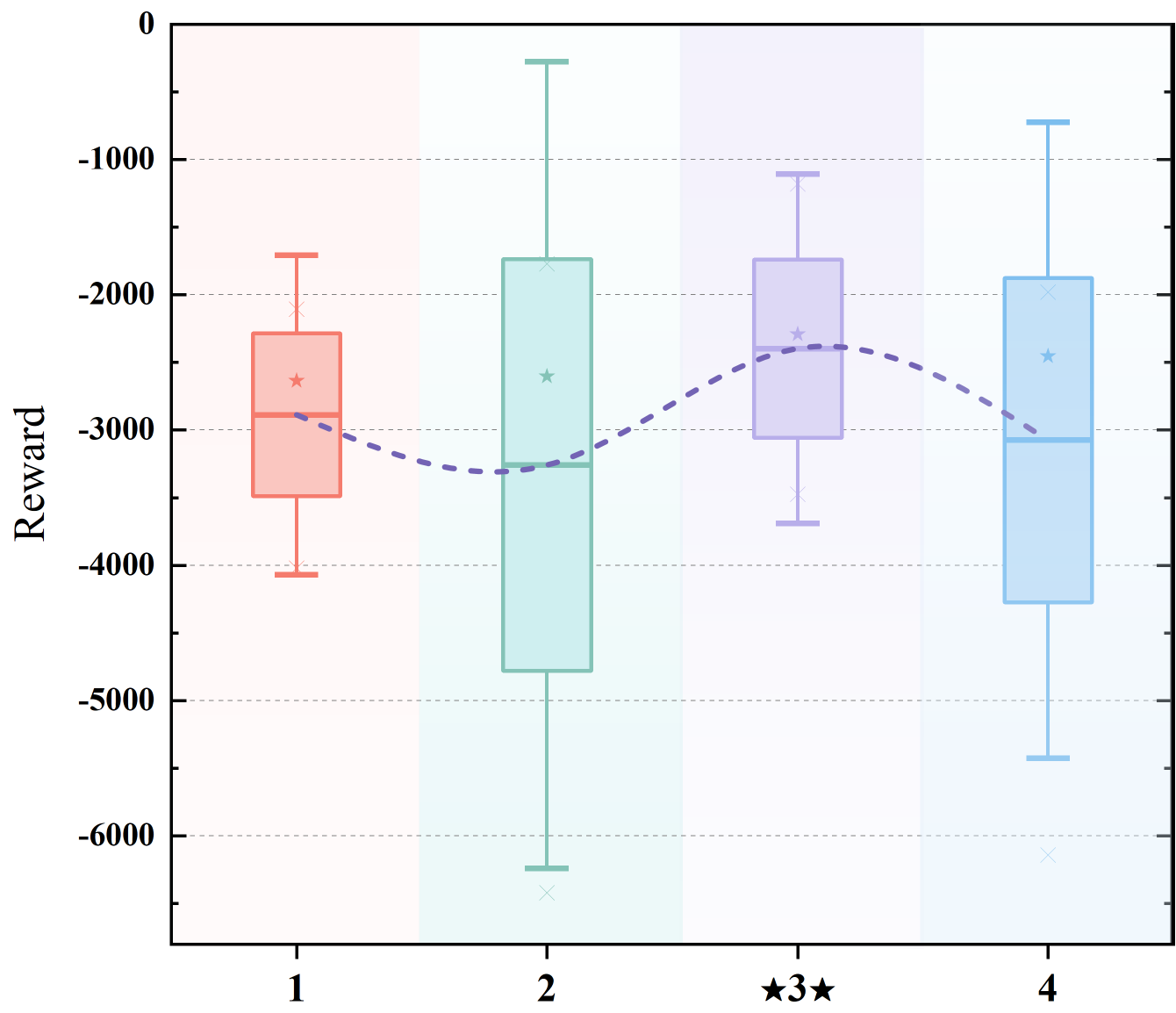} 
\caption{Impact of four different role configurations, including single role (\texttt{Executor}), two roles (\texttt{Commander}–\texttt{Executor}), three roles (\texttt{Commander}–\texttt{Coordinator}–\texttt{Executor}), and four roles (\texttt{Commander}–\texttt{Coordinator}–\texttt{Executor}–\texttt{Decoy}).}
\label{RoleNum}
\end{figure}
Overall, the three-role (\texttt{Commander} – \texttt{Coordinator} – \texttt{Executor}) architecture strikes the optimal balance between performance and stability within our two-phase LLM–MARL convergence framework, fully leveraging semantic decision making and reinforcement-driven exploration to achieve superior formation coverage and convergence stability.

\begin{figure}[t] 
\centering
\includegraphics[width=0.8\linewidth]{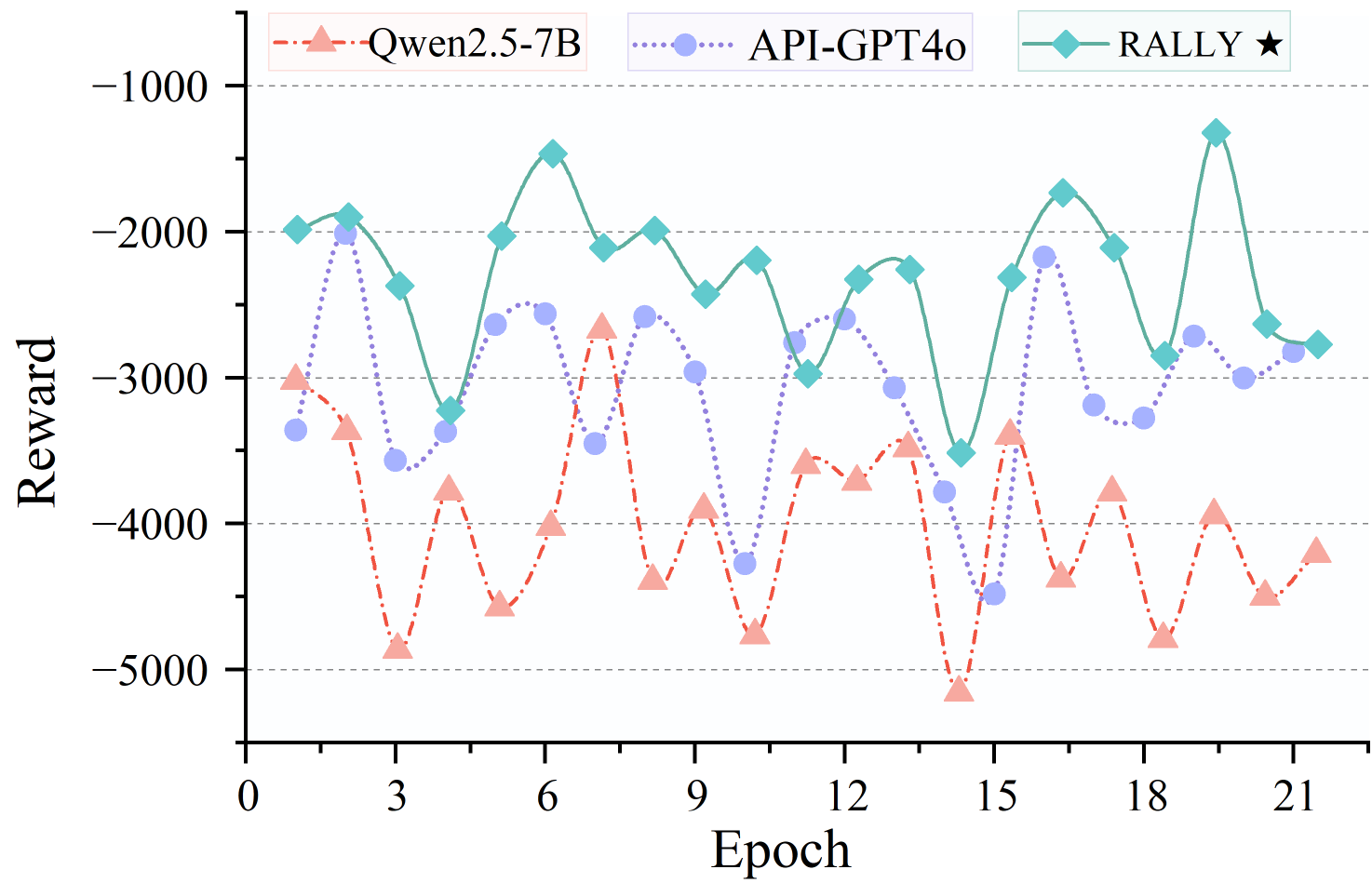}
\caption{Inference performance comparison after model fine-tuning.}
\label{FT}
\end{figure}
\begin{figure}[t] 
\centering
\includegraphics[width=0.8\linewidth]{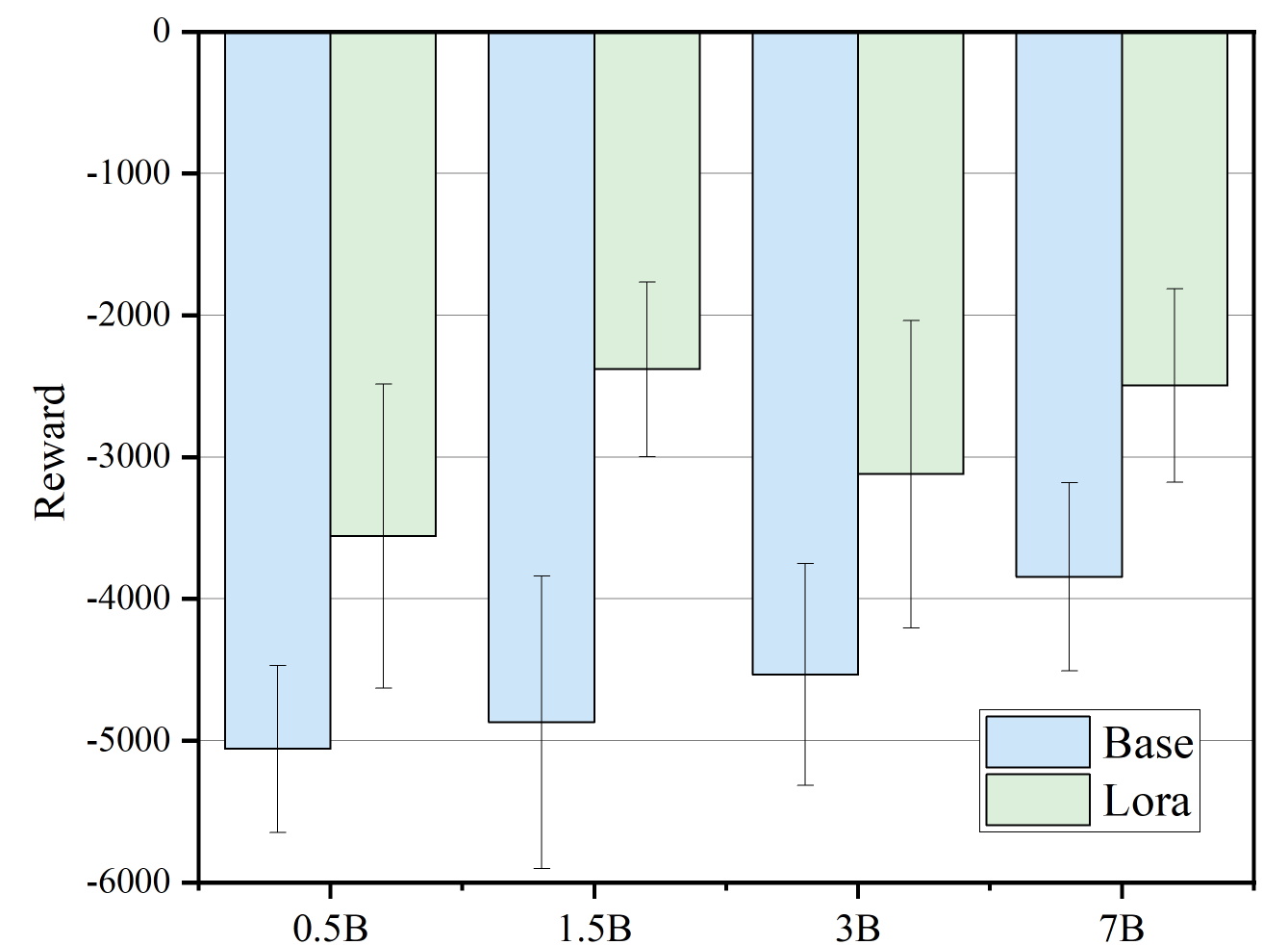}
\caption{Impact of model parameter scale.}
\label{ModelScale}
\end{figure}

\subsubsection{LLM Finetuning}
Fig. \ref{FT} present the
fine-tuning performance for RALLY and compare with other models. 
It can be observed from Fig. \ref{FT} that a non-fine-tuned Qwen2.5-$7$B base model delivers markedly lower performance, while due to its occasional illegal outputs (e.g., ``I suggest going to target point \#8, \#8.'' or ``region 8''), a direct calling of GPT-4o for online interaction possibly result in significant latency and network instability issues, degrading the performance. By contrast, RALLY, which fine-tunes Qwen2.5-7B base model under the dataset $\mathcal{B}_{\mathrm{fil}}$, harmonizes the high-quality inference of API-GPT-4o with the efficiency and stability of a smaller model.

\begin{table}[t]
    \centering
    \caption{Running performance of mainstream LLMs on an NVIDIA RTX 4090.}
    \label{tab:model}
    \renewcommand{\arraystretch}{1.25}
    \begin{tabular}{c|c|c|c}
    \hline
    \textbf{Model}  &\textbf{AIT} (s) &\textbf{MF} (GB) &\textbf{RO}  (GB)  \\ 
    \hline
    Qwen2.5-7B        & 15.39        &  15        & 15.7          \\
    Qwen2.5-3B         & 17.63        &  5.8       & 7.17          \\
    Qwen2.5-1.5B       & 14.48        &  2.9       & 4.13          \\
    Qwen2.5-0.5B       & 15.45        &  1.2       & 1.77          \\
    \bottomrule
    \end{tabular}
\end{table}

Next, we evaluate the performance after fine-tuning a LoRA-based family of Qwen2.5 models with varying parameter counts ($0.5$-B, $1.5$-B, $3$-B, and $7$-B) on the DS-CEFC task. Fig. \ref{ModelScale} summarizes the post-fine-tuning performance and demonstrates substantial performance gains. Furthermore, Table \ref{tab:model} reports the Average Inference Time (AIT), Memory Footprint (MF), and Runtime Overhead (RO) after running these models on an NVIDIA RTX 4090. It can be found that a Qwen2.5-$1.5$-B version strikes the balance by delivering robust decision quality with minimal inference overhead. 
\begin{figure}[!t]
\centering
\includegraphics[width=1.0\columnwidth]{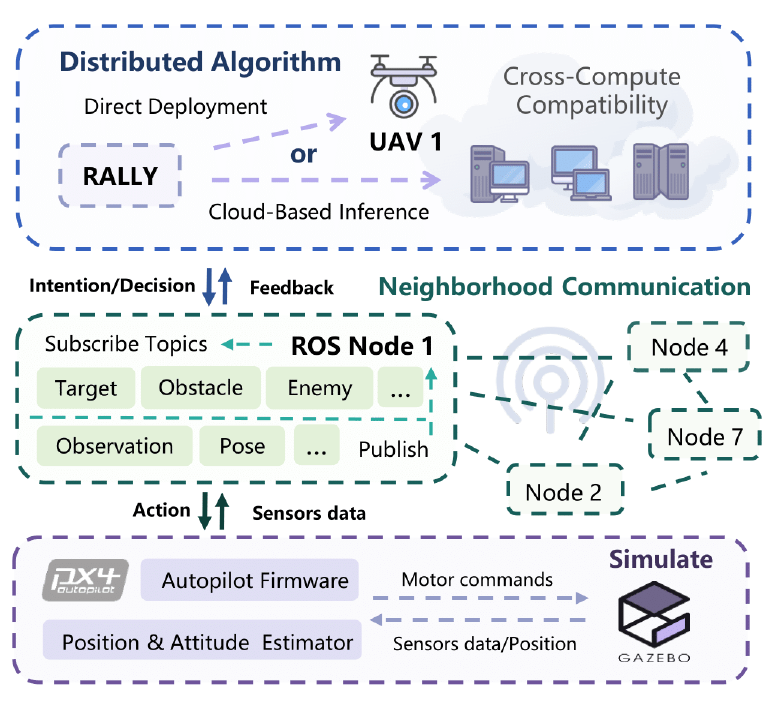}
\caption{Task overview in Gazebo Simulator for SITL.}
\label{sitl}
\end{figure}

\begin{figure*}[t!]
\centering
\includegraphics[width=1.0\textwidth]{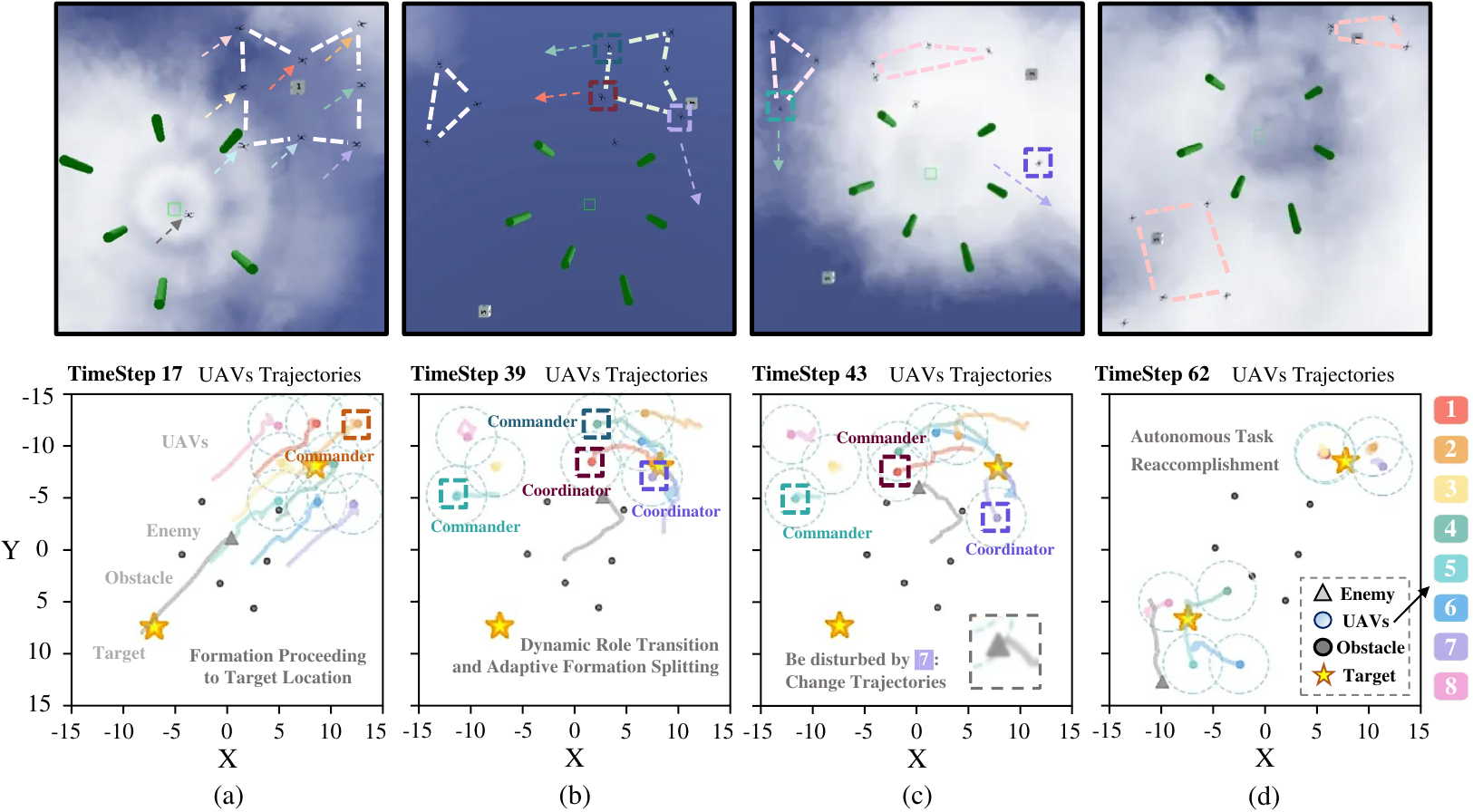}
\caption{Four typical cooperation scenarios in one episode of SITL. (a), (b), (c), and (d) stand for snapshots at high-level decision steps 17, 39, 43, and 62, respectively.}
\label{gazebo}
\end{figure*}

\begin{figure}[!t]
\centering
\includegraphics[width=1.0\columnwidth]{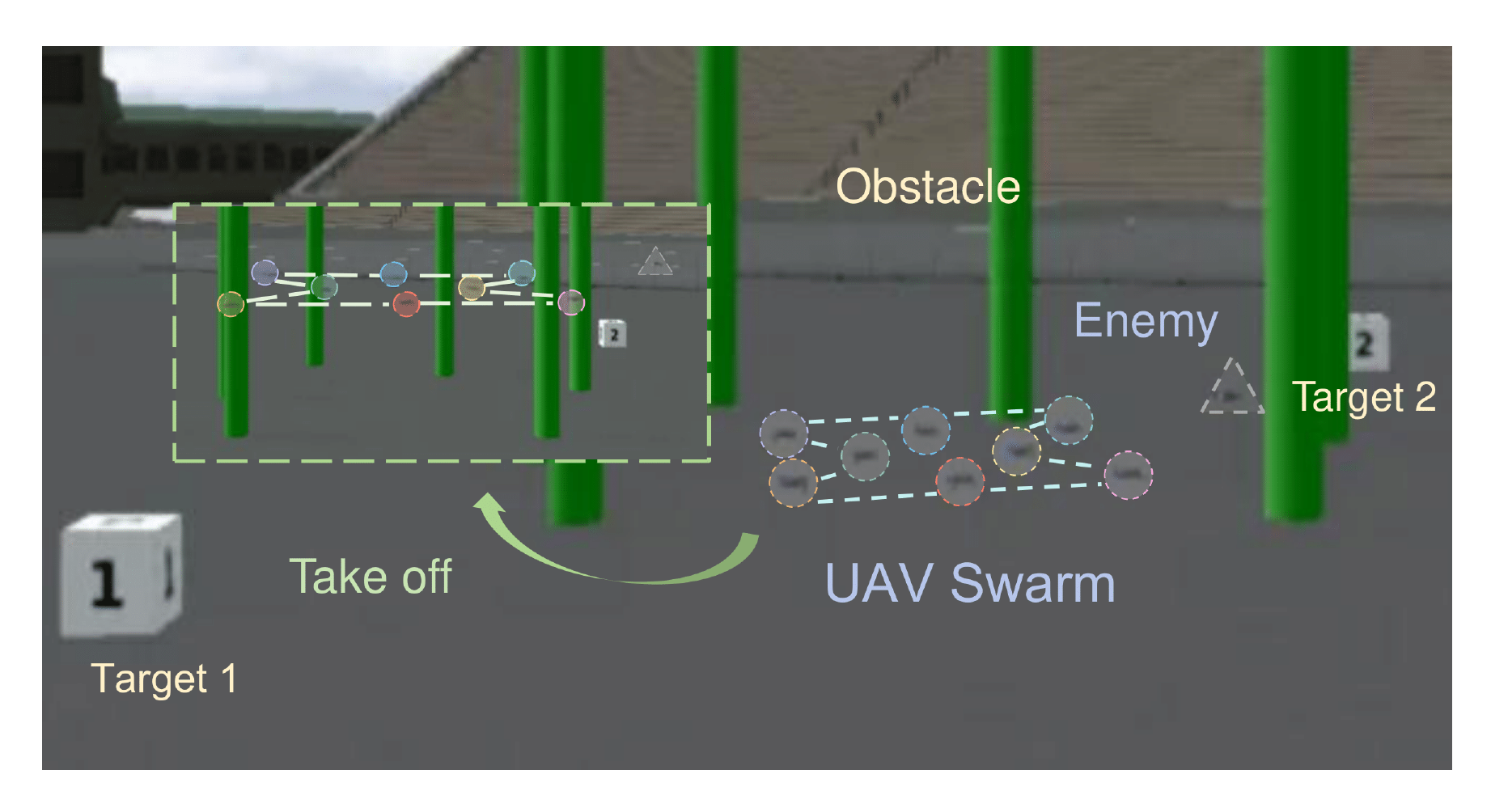}
\caption{Initialization of the Gazebo-ROS-PX4 simulation environment.}
\label{gazebointro}
\end{figure}
\subsection{Software-In-The-Loop Validation}
To assess RALLY's real-world viability, we integrate it into a ROS-based SITL environment featuring Gazebo-Classic and the PX4 autopilot, as illustrated in Fig. \ref{sitl}. Furthermore, each UAV follows standard quadrotor dynamics \cite{bresciani2008modelling}. 
Consistent with the mainstream MARL-based UAV studies, we assume fixed-altitude flight that ${u_z} \equiv 0$. 
Given the desired horizontal acceleration $\bm{u}=[{u_x}, {u_y}]$ from $\bm{\pi}_M$, the PX4 flight controller will uses a PID scheme\cite{meier2015px4} to compute thrust and angular rate commands, followed by a physics simulator which integrates the Newton–Euler equations\cite{bresciani2008modelling} to update each UAV's pose and dynamic state. 
Unlike prior open-source frameworks such as XT-Drone, our setup enforces fully distributed decision-making: each quadcopter node operates on local observations within a limited communication radius amid multiple obstacles, scoring zones, and predator–prey interactions. 
Concretely, UAV \#1 (and each peer) runs an independent off-board Python controller.
The high-level RALLY consensus module adopts the aforementioned $1.5$-B fine-tuned Qwen2.5 model, while the mid-level and low-level PX4 flight-control modules execute on a ground server to generate navigation commands.
Using MAVROS over UDP, each UAV publishes its state and sensory data (including adversary, obstacle, target, and neighbor information) and subscribes to receive pertinent updates. 
The desired target region, output by RALLY, is then converted into horizontal accelerations and broadcast to PX4 via ROS topics. 
PX4, connected to Gazebo through TCP, receives these acceleration demands, computes motor and actuator setpoints via its PID controllers, and returns updated UAV poses and sensor readings for the next simulation step. 
This tightly coupled loop in Fig. \ref{gazebointro} contributes to validating RALLY's distributed consensus and navigation performance under realistic quadcopter constraints.

Fig. \ref{gazebo} illustrates four representative consensus-building steps in the SITL environment, overlaid on each UAV's flight path derived from Gazebo-Classic. Each simulation time-step corresponds to one decision frame for consensus refinement. 
At the time-step 17 (Fig. \ref{gazebo}(a)), UAV \#2—being closest to Target \#2 and farthest from the enemy—assumes the \texttt{Commander} role and selects the upper-right scoring zone. All teammates comply and proceed toward Target \#2. 
At the time-step 39 (Fig. \ref{gazebo}(b)), the approaching enemy forces a reconfiguration: UAVs 3, 5, and 8 split off as a three-agent squad ($F_3$), with UAV \#5 promoted to \texttt{Commander} and guiding its group to Target \#1. The remaining five UAVs form a separate $F_5$ team; UAV \#4 takes on the \texttt{Commander} role and, alongside UAV \#1 acting as \texttt{Coordinator}, leads its squad toward the same goal. UAV \#7 selects to serve as \texttt{Coordinator}, deliberately positioning itself between the adversary and the cluster to divert attention and safeguard its peers. 
At the time-step 43 (Fig. \ref{gazebo}(c)), UAV \#7's unexpected directional shift effectively confuses the enemy's pursuit vector, causing it to veer off and granting the other drones a clear corridor to bypass the threat and reorient toward the next target. 
Finally, at the time-step 62 (Fig. \ref{gazebo}(d)), after successive rounds of LLM-driven intent generation and RMIX-guided role reassignment, both sub-clusters successfully evade the enemy and complete coverage of their respective scoring regions. 
This sequence confirms RALLY's ability to orchestrate dynamic role adaptation and robust distributed consensus in complex, adversarial scenarios.
\section{Conclusion and Future Works}
\label{sec:conclusion}
This paper introduces RALLY, an advanced LLM-MARL-integrated framework for UAV swarm control that combines LLM-driven semantic reasoning with MARL-based exploration.
By integrating autonomous intent understanding, dynamic role assignment, and decentralized consensus building, RALLY enables each UAV to interpret local observations, select functional roles, and collaboratively decide on navigation goals under communication constraints.
We validate RALLY via the MPE simulation environment and a high-fidelity Gazebo-ROS-PX4 SITL platform, demonstrating its superior task completion, collaborative effectiveness, and generalization compared to existing methods in the DS-CEFC scenario. 

Looking forward, we plan to optimize the lightweight LLM for faster inference and reduced communication latency, and enhance system robustness through more large-scale settings and advanced communication strategies.
In addition, we will address the possible local optima issue in CoT reasoning by exploring test-time training strategies and diversifying reasoning paths to improve reasoning diversity, reduce convergence to suboptimal solutions, and enhance generalization.
We also aim to investigate multimodal fusion and theoretical guarantees for rapid semantic consensus in larger UAV swarms. 
These efforts will pave the way for deploying intelligent, collaborative UAV systems in complex, resource-constrained missions.



\section*{Appendix}
In the Appendix, we provide the detailed prompts in Fig. \ref{AppendixPrompt} and give the reasoning sensitivity in Fig. \ref{llminput}.
\begin{figure*}[!tb]
\centering
\includegraphics[width=1\textwidth]{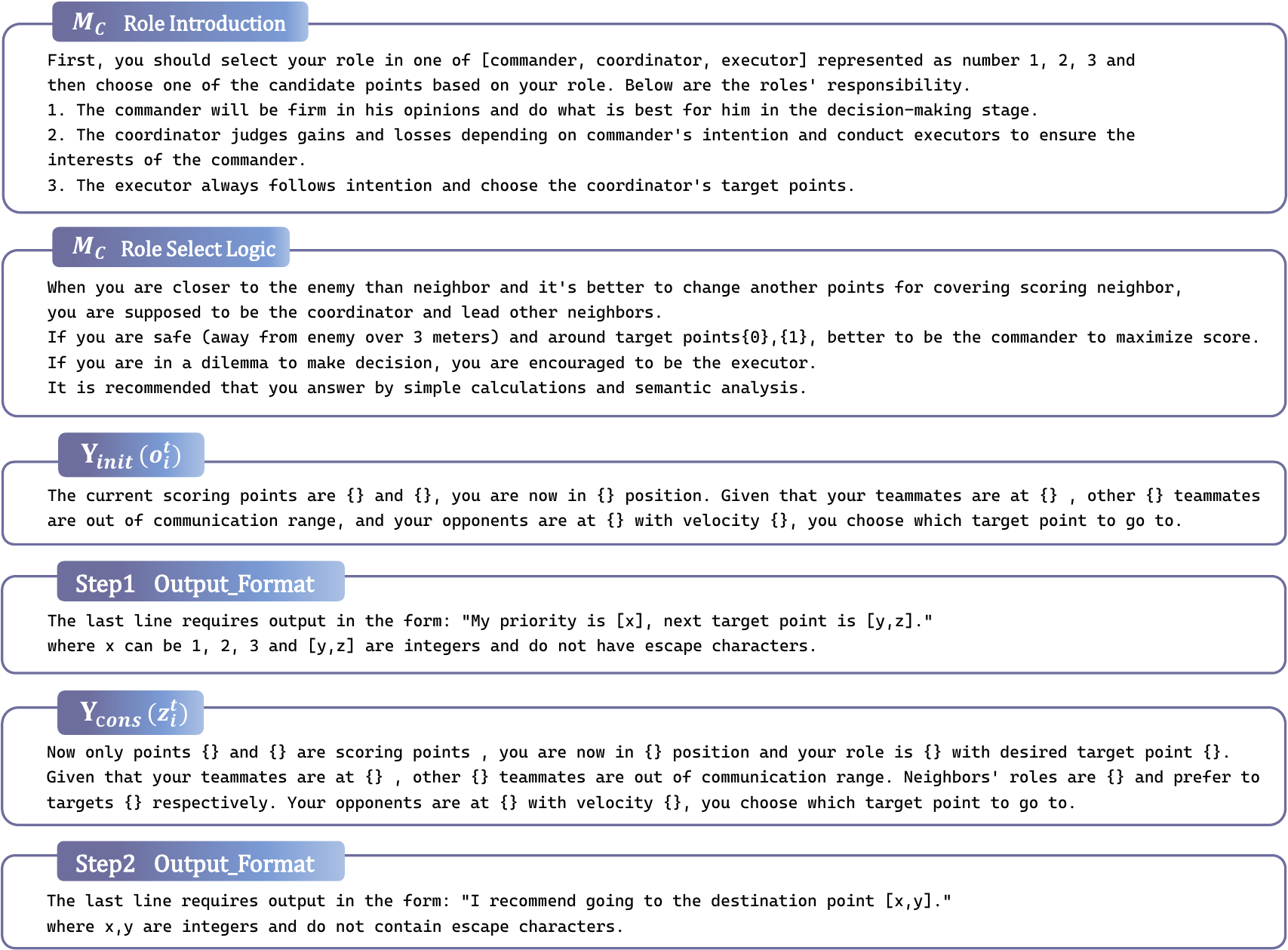}
\caption{Part of the prompts used in RALLY.}
\label{AppendixPrompt}
\end{figure*}

\begin{figure*}[!t]
\centering
\includegraphics[width=.85\textwidth]{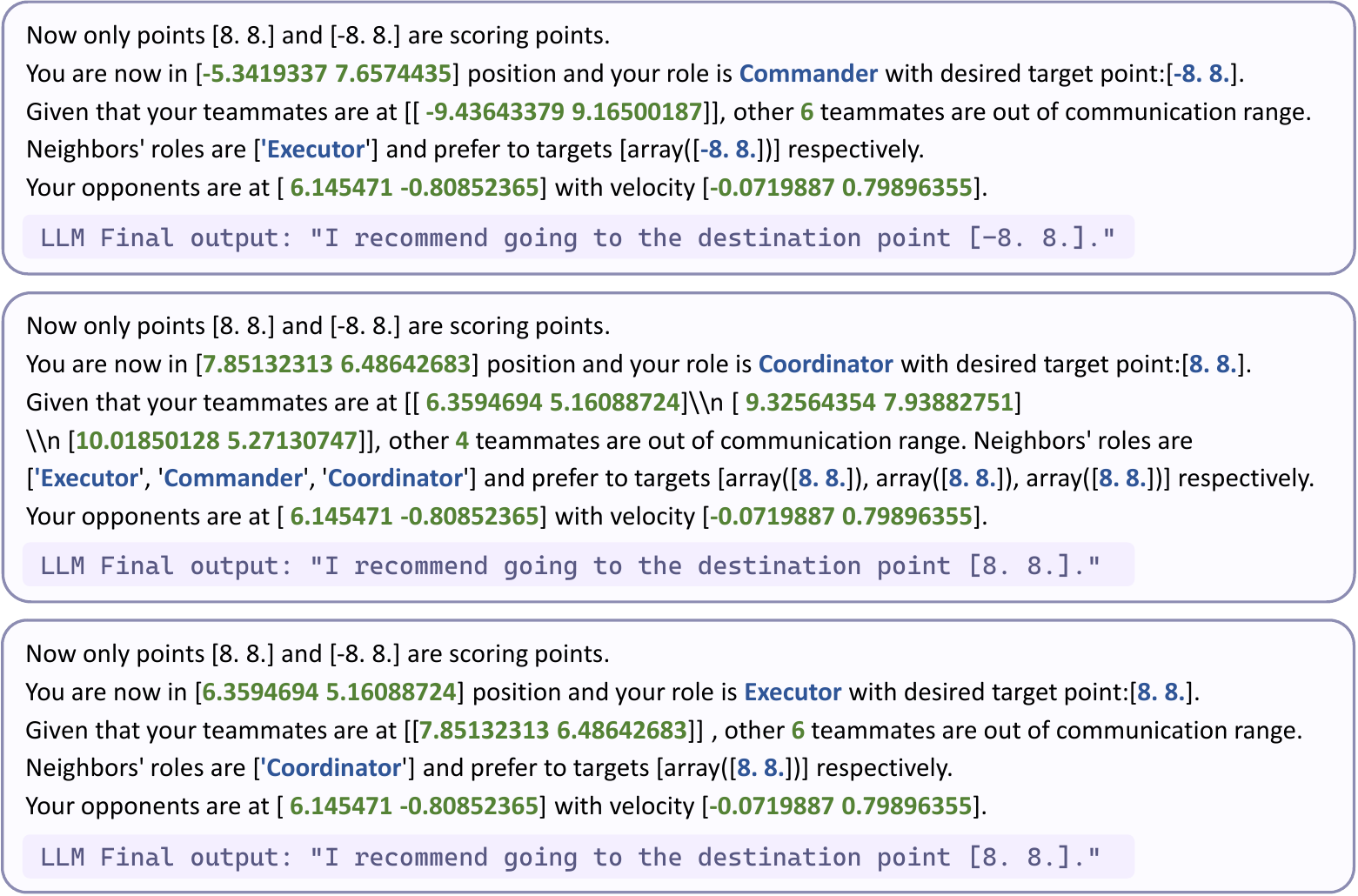}
\caption{Specific prompts and corresponding responses by different agents running their distributed models in parallel. For brevity, we have omitted the repeated structural prompts and CoT reasoning, while retaining only the key data-containing components of the input and output.}
\label{llminput}
\end{figure*}
\end{document}